\documentclass[11pt]{article}

\usepackage{amsfonts}
\usepackage{amsmath}
\usepackage{amssymb}
\usepackage{amsthm}
\usepackage[mathscr]{eucal}
\usepackage{bbold}
\usepackage{braket}
\usepackage{color}
\usepackage{dsfont}
\usepackage{framed}
\usepackage{jheppub}
\usepackage{mathtools}
\usepackage{physics}
\usepackage[normalem]{ulem}
\usepackage{tensor}
\usepackage{thmtools}
\usepackage{thm-restate}
\usepackage{ulem}

\definecolor{SHCcolor}{rgb}{0.9,0.1,0.1}

\definecolor{GHcolor}{rgb}{0.9,0.1,0.1}

\newcommand{\ignore}[1]{}

\newtheorem{nthm}{Theorem}[]

\newtheorem{nprop}{Proposition}

\newtheorem{ndefi}{Definition}

\renewcommand{\vec}[1]{\boldsymbol{\mathbf{#1}}}
\renewcommand{\(}{\left(}
\renewcommand{\)}{\right)}

\newcommand{\nd}{\quad\textrm{and}\quad}

\newcommand{\st}{\, \mid \,}
\newcommand{\where}{\quad\text{where}\quad}

\DeclareMathOperator{\codim}{codim}
\DeclareMathOperator{\id}{id}
\DeclareMathOperator{\vol}{vol}

\newcommand{\R}{\mathbb{R}}
\newcommand{\s}{\mathbb{S}}
\newcommand{\Z}{\mathbb{Z}}

\makeatletter\def\@fpheader{~}\makeatother
\newcommand{\ads}{\text{AdS}_n}
\newcommand{\adss}{\text{AdS}_n\times\s^k}
\newcommand{\sk}{\s^k}

\title{Bulk reconstruction of metrics with a compact space asymptotically}

\author{Sergio Hern\'{a}ndez-Cuenca and Gary T. Horowitz}

\affiliation{Department of Physics, University of California, Santa Barbara, CA 93106, USA}

\emailAdd{sergiohc@physics.ucsb.edu}
\emailAdd{horowitz@ucsb.edu}

\abstract{

Holographic duality implies that the geometric properties of the gravitational bulk theory should be encoded in the dual field theory. These naturally include the metric on dimensions that become compact near the conformal boundary, as is the case for any asymptotically locally $\text{AdS}_n\times\mathbb{S}^k$ spacetime. Almost all previous work on metric reconstruction ignores these dimensions and would thus at most apply to dimensionally-reduced metrics. In this work, we generalize the approach to bulk reconstruction using light-cone cuts and propose a prescription to obtain the full higher-dimensional metric of generic spacetimes up to an overall conformal factor. We first extend the definition of light-cone cuts to include information about the asymptotic compact dimensions, and show that the full conformal metric can be recovered from these extended cuts. We then give a prescription for obtaining these extended cuts from the dual field theory. The location of the usual cuts can still be obtained from bulk-point singularities of correlators, and the new information in the extended cut can be extracted by using appropriate combinations of operators dual to Kaluza-Klein modes of the higher-dimensional bulk fields.
}

\begin{document}


\maketitle


\section{Introduction}\label{sec:intro}
A central question in holography \cite{Maldacena_1999,Witten1998,Gubser:1998bc,Aharony_2000} is to understand how spacetime geometry emerges from the dual field theory. The standard discussions of entanglement wedge reconstruction do not address this since they depend on a choice of code subspace which represents small fluctuations about a given semiclassical bulk geometry (or perhaps a finite number of such geometries) \cite{Harlow:2018fse}. The idea of geometry emerging from entanglement \cite{VanRaamsdonk:2010pw} has led to various attempts to determine the bulk metric from measures of entanglement \cite{Balasubramanian:2013rqa,Balasubramanian:2013lsa,Myers:2014jia,Czech:2014wka,Headrick:2014eia,Czech:2014ppa,Czech:2015qta,Cao:2016mst}, in particular via the geometrization in the bulk of the von Neumann entropy of boundary regions \cite{Ryu_2006,Hubeny_2007}. Recently, it has been shown that a bulk geometry (if it exists) is uniquely determined by second order variations of the area of two-dimensional extremal surfaces anchored to a certain family of regions on the boundary \cite{Bao_2019}.\footnote{These areas correspond to the entanglement entropy of boundary regions when the bulk is four-dimensional.
In other dimensions, they are related to expectation values of Wilson loops in some cases, but their general holographic interpretation is not well understood \cite{Bao_2019,Bao:2019hwq,Maldacena:1998im}.}

In \cite{Engelhardt2016,Engelhardt2017} a very different approach toward reconstructing the bulk geometry was presented. This involves special cross-sections of the conformal boundary of an asymptotically Anti-de Sitter (AdS) spacetime. These cross-sections are called light-cone cuts, and can be thought of as the intersection of the past (or future) light cones of bulk points with the boundary.\footnote{Analogous cross-sections of null infinity in asymptotically flat spacetimes were first introduced in \cite{Newman:1976gc,Hansen:1978jz}, and shown to encode the conformal metric of such spacetimes in \cite{Kozameh:1983yu}. However, holography or the presence of internal spaces played no role in these discussions.} (A more precise definition will be given in the next section.) It was shown that knowledge of these light-cone cuts is sufficient to determine the conformal metric in the bulk, i.e. the metric up to an overall local rescaling, for most points causally related to the boundary.\footnote{An explicit implementation of the light-cone cut approach to bulk reconstruction was explored in \cite{Trevino:2017mik}. The light-cone cut formalism was also used in \cite{Engelhardt:2016vdk} to covariantize the notion of bulk depth and relate it to energy scales in the dual field theory.} It was further shown how to determine these light-cone cuts from singularities in certain time-ordered Lorentzian correlators in the dual field theory which originate precisely from bulk locality \cite{Gary:2009ae,Maldacena2015}. 

The results in \cite{Engelhardt2016,Engelhardt2017}, as well as those which employ entanglement entropy, apply to spacetimes which are asymptotically AdS. However, the most well studied models of holography require spacetimes to approach $\adss$. The goal of this paper is to extend the analysis of light-cone cuts to these more general spacetimes. (We will always assume $n>2$, since for $n=2$ the light-cone cut consists of isolated points and does not determine the conformal metric.) It is easy to see that a naive, straightforward attempt to apply light-cone cuts to spacetimes with compact extra dimensions will fail to determine the bulk geometry. However, we will show that there is a generalization of light-cone cuts that we call ``extended cuts'', that indeed determine the conformal metric of the full higher-dimensional spacetime. We will then propose a method for obtaining these extended cuts from the dual field theory.

The basic idea behind our extended cuts is the following. Every null geodesic which reaches the boundary of AdS approaches a fixed point on $\sk$. This is simply because a geodesic with asymptotic motion on the sphere acts like a massive particle in AdS and will not reach the boundary. Thus for every point on our light-cone cut, we can associate a point on $\sk$. Our extended cut is just the original light-cone cut $C(p)$ together with a map $ C(p) \to \sk$ specifying the asymptotic location on  $\sk$ of the null geodesics from $p$ to $C(p)$. We show in Section \ref{sec:exlcc} that this map is precisely the extra information that is needed to  reconstruct the full bulk conformal metric of generic spacetimes (if the spacetime has symmetries that asymptotically act only on the internal space, this approach may fail). In Section \ref{sec:bdydata}, we propose a way to determine the extended cuts from the dual field theory, using correlation functions involving the operators dual to the Kaluza-Klein modes of the higher-dimensional bulk field. Our approach does not use any particular property of $\sk$ and should work equally well for a general compact internal space.

To completely determine the bulk geometry, one also needs to know the conformal factor. This remains an open problem in general, however it is known how do to this for some special asymptotically AdS spacetimes \cite{Engelhardt2017}.

\section{Review of Light-Cone Cuts}\label{sec:review}
In this section we review the construction in \cite{Engelhardt2016,Engelhardt2017} for obtaining the conformal metric from the dual field theory. Two metrics $g$ and $\bar{g}$ are conformally related if there is a positive function $f$ such that $\bar g_{\mu\nu} = f^2 g_{\mu\nu}$. Points in spacetimes with conformally related metrics clearly have the same light cone, but one does not need to know the entire light cone (or even an open subset of it) to determine the conformal metric at those points; a sufficient number of null vectors will do. This can be seen as follows. In a $D$-dimensional spacetime, take $D$ linearly independent null vectors $\ell_i$ at a point $p$.
Since the $\ell_{i}$ all have zero norm, the conformal metric at $p$ is fixed by their inner products. To determine them, take a new collection of null vectors, $\eta_k$, and expand them in terms of the null basis $\ell_i$: 
\begin{equation} \eta_{k} = \sum\limits_{i} M_{ki}\ell_{i}.
\end{equation} 
Using the fact that each $\eta_{k}$ has zero norm, we obtain a set of algebraic equations for the inner products $\ell_i \cdot \ell_j$:
\begin{equation}0= \eta_{k}\cdot \eta_{k} = \sum_{i,j=1}^{D} M_{ki}M_{kj} (\ell_{i}\cdot \ell_{j}) \qquad {\rm no\ sum\ on \ }k.
\end{equation}
While it is not always true that such equations have a solution, we are guaranteed a solution here precisely because these equations describe a Lorentzian metric which by construction exists. By choosing at least $D(D-1)/2$ vectors $\eta_k$, the solution will be unique up to an overall constant rescaling of all inner products. This determines the conformal metric at $p$. Repeating this local construction at each point in a spacetime region $U$ determines the conformal metric on $U$.

Our goal is to determine these null vectors at $p$ from boundary data. Due to gravitational lensing, the light cone of a bulk point $p$ can develop caustics. When this happens, some null geodesics reach points that are timelike related to $p$. Since we want boundary points that are null-related to $p$ we proceed as follows.

Let $(M,g)$ be an asymptotically locally AdS spacetime (without compact extra dimensions) with conformal boundary $\partial M$, and denote its conformal compactification by $(\bar{M},\bar{g})$. Recall that the causal past $J^{-}(p)$ of a point $p\in M$ is the set of points in $M$ which can be reached by a past-directed causal curve starting at $p$. $J^{+}(p)$ is defined similarly with ``future" replacing ``past". A spacetime is said to be AdS-hyperbolic if there exist no closed causal curves and for any two points $p,q\in M$, the set $J^+(p)\cap J^-(q)$ is compact in the conformally compactified spacetime $\bar{M}$ \cite{Wall_2014}. We will assume our spacetime is $C^2$ differentiable, maximally extended, connected, and AdS-hyperbolic. The future/past light-cone cut $C^\pm(p)$ of a point $p\in M$ is defined as the intersection of the boundary of the causal future/past of $p$, $\partial J^\pm(p)$, with the conformal boundary $\partial M$, i.e.
\begin{equation}\label{eq:cutdef}
C^\pm(p) \equiv \partial J^\pm(p) \cap \partial M.
\end{equation}
This is illustrated in Fig. \ref{fig:lcc}. We will use $C(p)$ to denote either the future or past cut of a bulk point $p$. Light-cone cuts are not differentiable everywhere since they can have cusps due to caustics. However, it can be shown that the cusps form a set of measure zero within the cut (cf. Proposition \ref{prop:zeromes} in Section \ref{ssec:geocuts}).

\begin{figure}
    \centering
    \includegraphics[height=.4\textheight]{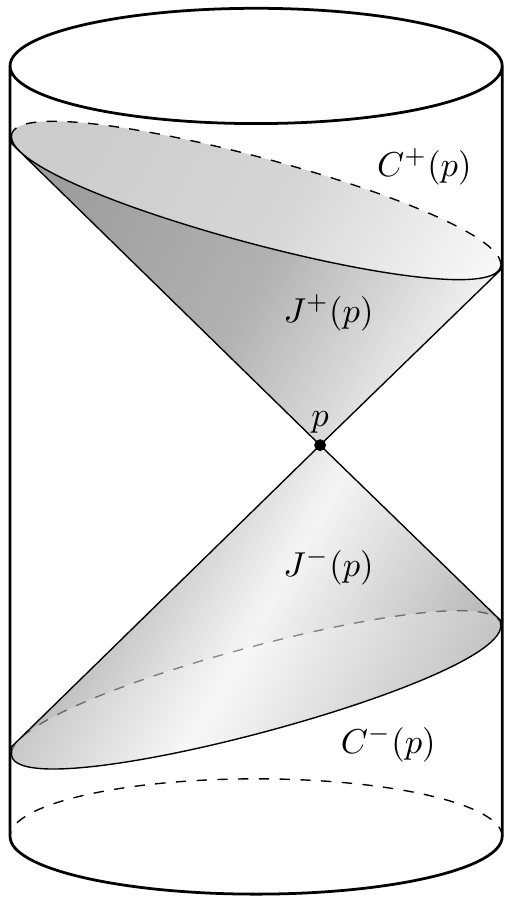}
    \hspace{50pt}
    \includegraphics[height=.4\textheight]{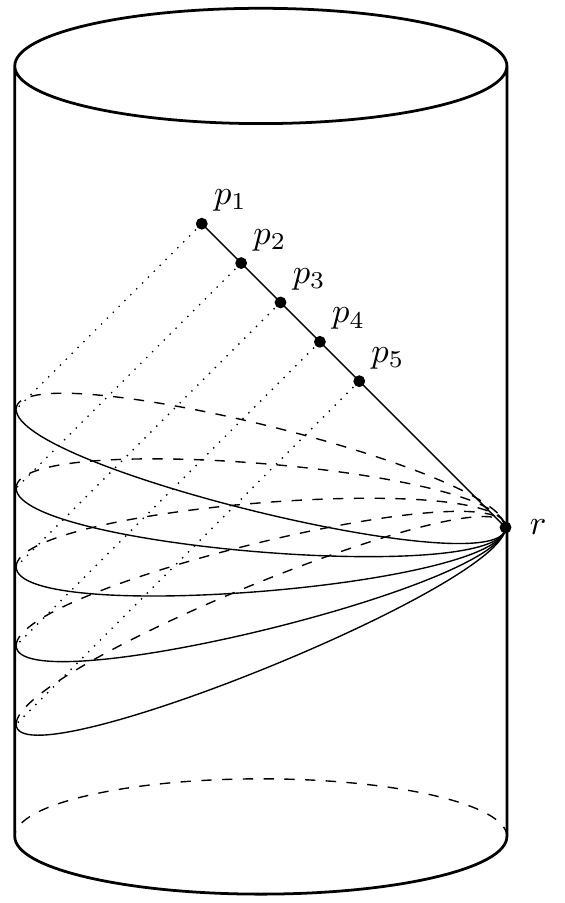}
    \caption{In the left figure, the shaded region illustrates the future and past light cones $\partial J^\pm(p)$ of a bulk point $p\in M$ in causal contact with the boundary in an asymptotically AdS spacetime $M$. Their intersections with the conformal boundary $\partial M$ define the future and past light-cone cuts $C^\pm(p)$, which are complete spatial slices of $\partial M$ (cf. properties \ref{complete} and \ref{1to1}). The right figure shows how a sequence of null-related bulk points $\{p_i\in M\}$ corresponds to a set of light-cone cuts which all intersect at a single point $r\in \partial M$ (cf. property \ref{null}).}
    \label{fig:lcc}
\end{figure}

It was shown in \cite{Engelhardt2016} that light-cone cuts satisfy the following properties:
\begin{enumerate}
\item\label{complete} $C(p)$ is a complete spatial slice of the conformal boundary.

\item\label{1to1} There is a one-to-one, onto map from past light-cone cuts to points in the future of the boundary, even inside black holes. (A similar statement holds for future cuts.)

\item\label{open} Two distinct cuts cannot agree on an open set.

\item\label{null} If $C(p)$ and $C(q)$ intersect at precisely one point, and both cuts are $C^1$ at this point, then $p$ and $q$ are null-separated.

\end{enumerate}

Using these properties, it is easy to construct the bulk conformal metric given the light-cone cuts. Property \ref{1to1} says that the set of past cuts represents all points to the future of the boundary. Property \ref{null} says that given a light-cone cut $C(p)$, the set of cuts $C(q)$ which are tangent to $C(p)$ at a regular point $r \in C(p)$ represents  a null curve passing through $p$, as illustrated in Fig. \ref{fig:lcc}. Repeating this for $D(D+1)/2$ cut points $r$ allows one to reconstruct the conformal metric at $p$.\footnote{One needs $D$ points for the basis vectors $\ell_i$, and $D(D-1)/2$ for the null vectors $\eta_k$ used to determine the inner products.} It is clear that a basis of null vectors $\ell_i$ at $p$ can be obtained this way, since the light-cone cut $C(p)$ enables one to reconstruct an open subset of the light cone at $p$.

The second half of the construction is a procedure for determining $C(p)$ from the dual field theory without using the bulk geometry. This is achieved using the notion of bulk-point singularities, first argued for in \cite{Gary:2009ae} and later studied in \cite{Maldacena2015}. Given $D$ boundary points in a $D$-dimensional spacetime, the only subset of $M$ which can be null-related to all of them are individual points. It was shown in \cite{Gary:2009ae,Maldacena2015} that a time-ordered Lorentzian $(D+1)$-point correlator on the boundary of AdS is singular when there exists a momentum-preserving scattering point in the bulk that is null-related to all of them (i.e. if one can draw a position-space Landau diagram with null lines in the bulk).\footnote{Even though a single bulk point $p$ can be fixed by the condition that it is null-related to $D$ boundary points, one needs at least one extra point in the correlator to ensure that energy-momentum is conserved at $p$.}  This is the case if, for example, one chooses two points in the past cut $x_1, x_2\in C^-(p)$, and $D-1$ points in the future cut $x_i\in C^+(p)$ of a bulk point $p$, in a manner similar to Fig. \ref{fig:bpc}. Then, physically, high energy quanta from $x_1$ and $x_2$ can scatter at $p$ conserving energy-momentum and send high energy quanta to the remaining $x_i$ in the future, which results in a singular correlator. In special cases, only derivatives of the correlator will diverge and the correlator itself may remain finite. However, for most operators, the correlator itself will diverge and we will use such operators below.

\begin{figure}
    \centering
    \includegraphics[height=.4\textheight]{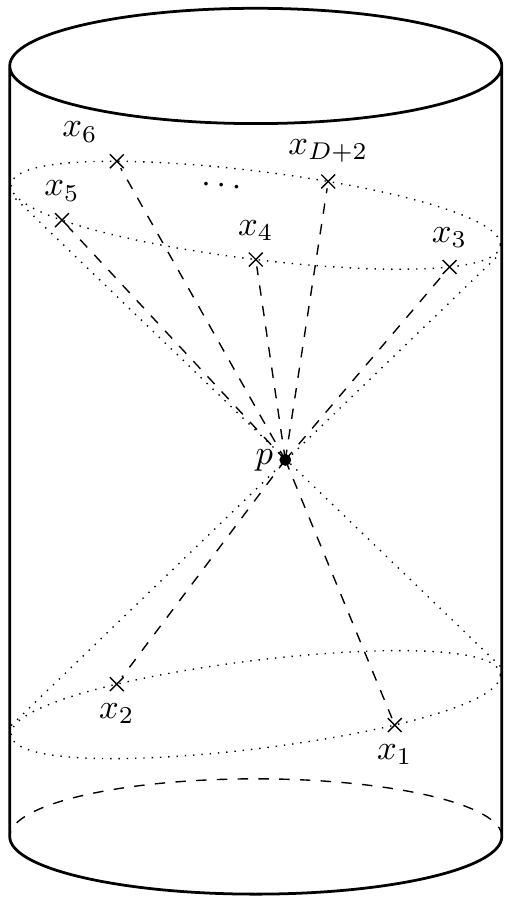}
    \caption{Position-space Landau diagram for a boundary correlator with a bulk-point singularity from $p\in M$ used to obtain light-cone cuts from the dual field theory. For a $D$-dimensional bulk, the $D$ boundary points in the future already specify $p$ as the unique bulk point null-related to all of them. The two points in the past can be rotated around maintaining momentum conservation at the interaction vertex (and hence the divergence in the correlator) to trace out the light-cone cut.}
    \label{fig:bpc}
\end{figure}

To use this to find the light-cone cuts we need two modifications. First, we consider correlation functions in certain excited states, not the ground state, so the dual spacetime is only asymptotically locally AdS and not pure global AdS. Second, we consider $(D+2)$-point correlators, with two points $x_1$ and $x_2$ in the past and $D$ points {$x_3,\dots,x_{D+2}$} in the future (see Fig. \ref{fig:bpc}). In this case, if there is a bulk point $p$ null-related to all the $x_i$ to the future, it will remain fixed if we move the ones in the past. Starting with a configuration of points where the correlator diverges, we can thus move $x_1$ and $x_2$ in a coordinated manner keeping the correlator singular to trace out the past cut of $p$.\footnote{One could actually work with $D+1$ points and still move one vertex in a limited way to trace out part of the light-cone cut, but one has more freedom to trace out the entire cut by adding an additional point. One must also minimize the time difference between the points in the past and future to avoid caustics along the null geodesics from the bulk point to the boundary.}


\section{Extended Light-Cone Cuts}\label{sec:exlcc}

In this section we extend the discussion of light-cone cuts to spacetimes that have a compact space asymptotically such as $\adss$. The presence of this compact space implies that most of the null geodesics on the light cone of a bulk point $p$ end up crossing other null geodesics and entering the interior of $J(p)$. Only a small subset of these null geodesics stay on the boundary of $J(p)$ and form the light-cone cut. To illustrate this, consider the three-dimensional flat spacetime $ds^2 = -dt^2 + dz^2 + d\chi^2$, with $\chi$ periodically identified. Starting at any point $p$, all null geodesics with $\dot\chi \ne 0$ will go around the $\s^1$ and reach points timelike related to $p$. The only ones that stay on $\partial J(p)$ are those with $\dot\chi = 0$. More generally, for spacetimes locally asymptotic to $\adss$, the light-cone cut has bulk codimension $k+2$ rather than $2$. This means that one cannot recover an open subset of the light cone of a bulk point $p$. Fortunately, as reviewed above, one does not need an open subset of the light cone to recover the conformal metric at $p$. All one needs is a basis of null vectors and some additional null vectors. As we discuss below, this can be obtained in generic spacetimes from a simple generalization of the light-cone cut.

For asymptotically locally $\adss$ spacetimes, one way to understand the reduction in the size of the light-cone cut is by noting that the conformal boundary of $\adss$ is degenerate, in the sense that it is codimension $k+1$ rather than $1$ \cite{taylorrobinson2000holography}. Indeed, the $\sk$ factor of the direct product shrinks to zero size and leaves a boundary manifold $\partial M$ which is locally isometric to the conformal boundary of just the $\ads$ part.

The presence of a degenerate boundary turns out to invalidate most results proven in \cite{Engelhardt2017,Engelhardt2016}. Fortunately, it is possible to recover them with appropriate generalizations of the framework. To motivate the solution, let us first understand the complications that arise when the boundary is degenerate. In particular, consider the following two results from \cite{Engelhardt2016} (cf. properties \ref{open} and \ref{null} reviewed in Section \ref{sec:review}) and counterexamples to them already in the simple case of global $\adss$:

\begin{itemize}
	\item \emph{$C(p) \cap C(q)$ contains a nonempty open set if and only if $p = q$}: For any two points $p$ and $q$ on $\adss$ with the same global coordinates on the AdS part one has $C(p) = C(q)$, even if they have different coordinates on the sphere. More precisely, thinking of the compactification space $\sk$ as a fiber of a trivial bundle $\pi : \adss \to \ads$, this means that $C(p)=C(q)$ for any $p,q\in\adss$ with the same base space point $\pi(p)=\pi(q)$, implying that light-cone cuts do not distinguish points on the fibers.
	\item \emph{If $C(p)$ and $C(q)$ intersect at precisely one point, and both cuts are $C^1$ at this point, then $p$ and $q$ are distinct and null-separated}: To falsify this claim, consider an arbitrary point $p\in\adss$ and another null-separated point $q\in\partial J(p)$ such that the null-geodesic between $p$ and $q$ reaches $\partial M$ at some point $r\in C(p)\cap C(q)$. It is easy to see that their light-cone cuts will indeed intersect precisely only at $r$, and that both cuts will be $C^1$ at this point (since the spacetime is pure $\ads$). Now take another point $\tilde{q}$ which is at the same AdS location as $ q$, but at a different point on the sphere. Since the metric on the sphere is Euclidean, $p$ and $\tilde{q}$ will be spacelike-separated. But from the counterexample to the previous claim, one still has $C(\tilde{q})=C(q)$. Altogether, this shows that $C(p)$ and $C(\tilde{q})$ intersect at precisely one point, both cuts are $C^1$ at this point, but $p$ and $\tilde{q}$ are spacelike-separated, thus contradicting the statement above.
\end{itemize}

As anticipated, the existence of these counterexamples can be traced back to the fact that the light cone $\partial J(p)$ of a bulk point $p\in M$ degenerates asymptotically in essentially the same way the conformal boundary does. More precisely, suppose a boundary observer wanted to resolve the compact dimensions by introducing a regulated boundary $\partial M_\epsilon$ at a finite UV cutoff $0<\epsilon\ll1$, with $\lim_{\epsilon \to 0} \partial M_\epsilon = \partial M$. On $\partial M_\epsilon$, the dimensions of $\sk$ are restored and one has $\codim \partial M_\epsilon = 1$, the dimensionality only dropping by $k$ in the strict limit $\epsilon\to0$. Similarly, intersecting $\partial J(p)$ with the regulated boundary $\partial M_\epsilon$, one sees that the corresponding regulated light-cone cut $C_\epsilon(p) = \partial J(p) \cap \partial M_\epsilon$ is now bulk-codimension $2$, the dimensionality only decreasing by $k$ in the strict limit $\epsilon\to0$. 

Crucially, under the pertinent assumptions, all results proven in \cite{Engelhardt2017,Engelhardt2016} apply now to regulated light-cone cuts. However, because the dual field theory does not gain any dimensions, we need to find a way to retain this information in the limit $\epsilon\to0$. Unsurprisingly, this will require supplementing the standard cuts $C(p)$ with some information from $C_\epsilon(p)$. Precisely how the $\epsilon\to0$ limit of $C_\epsilon(p)$ can be used to extend $C(p)$ sufficiently for the light-cone cut reconstruction to succeed is the subject of this section.

\subsection{Asymptotics of spacetimes with degenerate boundaries}\label{ssec:asymp}

The first step is to have an elementary understanding of how null geodesics behave asymptotically in spacetimes with an internal space. Henceforth, the bulk spacetime $M$ is assumed to be asymptotically locally isometric to $\adss$, whose metric in global coordinates reads
\begin{equation}\label{eq:gmetric}
g = - f(r) dt^2 + f(r)^{-1} dr^2 + r^2 d\Omega_{n-2}^2 + \ell^2 d\Omega_{k}^2 \where f(r) = 1 + \frac{r^2}{\ell^2}.
\end{equation}
Here $\ell$ is the radius of curvature of $\ads$, and the shorthand $\Omega_d$ is used to collectively refer to all coordinates on $\s^d$. Define dimensionless time $\tau$ and radial $\rho$ coordinates via $\tau = t/\ell$ and $ r = \ell\tan\rho $, so that \eqref{eq:gmetric} becomes
\begin{equation}\label{eq:bargfirst}
g = \frac{\ell^2}{\cos^2\rho} \( - d\tau^2 + d\rho^2 + \sin^2\rho \, d\Omega_{n-2}^2 + \cos^2\rho \, d\Omega_{k}^2 \).
\end{equation}
Since null geodesics are only sensitive to the causal structure, which depends just on the conformal class of the metric, consider a Weyl rescaling $ g \mapsto \bar{g} = \omega^2 g $, with $\ell \omega = \cos\rho$. This gives
\begin{equation}\label{eq:barg}
\bar{g} = - d\tau^2 + d\rho^2 + \sin^2 \rho \, d\Omega_{n-2}^2 + \cos^2 \rho \, d\Omega_{k}^2,
\end{equation}
which is simply time cross $\s^{n+k-1}$. Noting that the conformal boundary $\partial M$ corresponds to the limit $\rho\to\pi/2$, it is now evident how the induced metric on $\sk$ degenerates in the strict asymptotic limit. In fact, this is no different from the way in which the metric degenerates at the origin $\rho=0$ in these coordinates. More explicitly, letting $\rho = \frac{\pi}{2} - \epsilon$ with $0<\epsilon\ll1$ and expanding locally in a neighborhood of $\partial M$, one finds
\begin{equation}\label{eq:gepsilon}
\bar{g} = - d\tau^2 +(1-\epsilon^2/2) \, d\Omega_{n-2}^2 + d\epsilon^2 + \epsilon^2 \, d\Omega_{k}^2 + \mathcal{O}\(\epsilon^4\).
\end{equation}
The $(\epsilon,\Omega_k)$ sector above provides a convenient chart on the space orthogonal to $\partial M$. The $\ads$ sector of \eqref{eq:barg} takes the familiar form of one half of the Einstein static universe, and the metric induced on $\partial M$,
\begin{equation}\label{eq:coordbdy}
\bar{g}_{\partial M} = - d\tau^2 + d\Omega_{n-2}^2,
\end{equation}
reveals the usual boundary topology $\R\times\s^{n-2}$ of conformally compactified AdS spacetimes.

The leading behavior of null geodesics in $M$ near the conformal boundary can be extracted from $g$ in \eqref{eq:bargfirst} in the limit $\rho\to\pi/2$. Since null geodesics are conformally invariant, we can actually work with \eqref{eq:barg}. Let $\gamma$ be a null geodesic curve with affine parameter $\lambda$ and tangent vector field $N = \dot{\gamma}$. The Killing symmetries of \eqref{eq:barg} give rise to several conserved quantities along $\gamma$. If we choose coordinates on the spheres so that the geodesic is moving in the $\varphi$ direction on $\s^{n-2}$ and $\psi$ direction on $\s^k$, then we get the following conserved charges:\footnote{These are only conserved charges in global $\adss$, and will not actually be conserved along $\gamma$ on $M$ in general. More appropriately, these quantities should be thought of as the asymptotic charges carried by $\gamma$ as it reaches $\partial M$.}
\begin{equation}
E = \dot{\tau},\qquad L_{n-2} = \sin^2\rho\,\dot{\varphi} \qquad\text{and}\qquad L_{k} = \cos^2\rho\,\dot{\psi},
\end{equation}
One can fix an arbitrary overall factor in $N$ by setting $E=\pm1$, where the sign determines the time orientation. The general asymptotic form of $N$ can thus be written
\begin{equation}
N^a = \pm (\partial_\tau )^a + \dot{\rho} \, (\partial_\rho)^a + \frac{L_{n-2}}{\sin^2\rho} ( \partial_{\varphi})^a + \frac{L_{k}}{\cos^2\rho} ( \partial_{\psi})^a,
\end{equation}
where the null condition $N^2=0$ constrains $\rho$ to obey
\begin{equation}\label{eq:dotrhosqd}
\dot{\rho}^2 = 1 - \frac{L_{n-2}^2}{\sin^2\rho} - \frac{L_{k}^2}{\cos^2\rho}.
\end{equation}
The limit $\rho\to\pi/2$ in \eqref{eq:dotrhosqd} makes it immediately clear that $\gamma$ can only reach $\partial M$ if $L_k=0$. This means that null geodesics only reach the conformal boundary if they approach a fixed point on $\s^k$ at infinity. This is easily understood from the perspective of Kaluza-Klein reduction, where a non-zero $L_k$ would physically correspond to a massive test particle on the dimensionally-reduced spacetime, which of course cannot reach the conformal boundary.

Expanding about $\partial M$ as in \eqref{eq:gepsilon}, the asymptotic form of $N$ becomes
\begin{equation}\label{eq:asymN}
N^a = \pm (\partial_\tau)^a - \sqrt{1-L_{n-2}^2} \, (\partial_{\epsilon(\Omega_{k})})^a + L_{n-2} \( \partial_{\varphi}\)^a +\mathcal{O}\(\epsilon^2\),
\end{equation}
where the notation $\partial_{\epsilon(\Omega_{k})}$ is introduced to make it explicit that the direction of the radial vector $\partial_{\epsilon}$ on the $(\epsilon,\Omega_k)$ sector is parameterized by the angular coordinates $\Omega_k$ on the asymptotic $\sk$, like in ordinary spherical coordinates. The corresponding parametric form of its asymptotic integral curve is thus, to leading order,

\begin{equation}\label{eq:paracurve}
{\gamma(\epsilon) = \( \tau^\infty \mp \epsilon ,\, \frac{\pi}{2} - \sqrt{1-L_{n-2}^2} \, \epsilon,\, \varphi^\infty - L_{n-2} \,\, \epsilon ,\, \Omega_{k}^\infty \) + \mathcal{O}\(\epsilon^2\),}
\end{equation}
where coordinates with superscripts $\infty$ denote asymptotic values and $\partial M$ is reached at $\epsilon=0$. Note that the limiting $\Omega_k^\infty$ will always be well defined despite the fact that the spherical coordinate system $(\epsilon,\Omega_k)$ degenerates at its origin $\epsilon=0$. In particular, $L_k=0$ implies that $\partial_{\epsilon(\Omega_{k})} = \partial_{\epsilon(\Omega_{k}^\infty)}$ in \eqref{eq:asymN}.

\begin{figure}
    \centering
    \includegraphics[width=.7\textwidth]{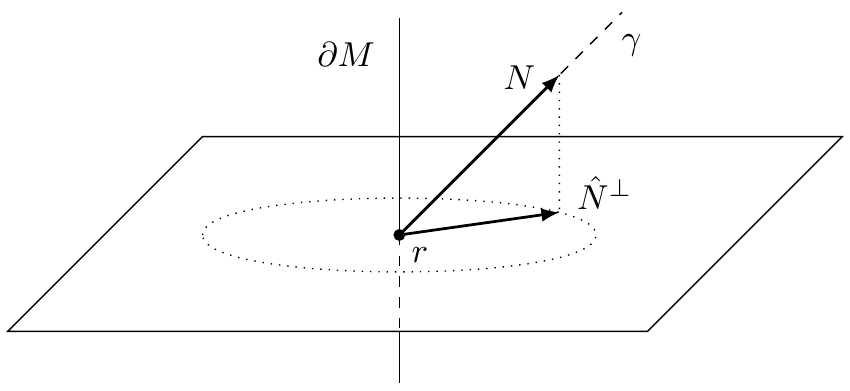}
    \caption{{Illustration of the tangent space normal to the conformal boundary at some $r\in\partial M$. The vertical line represents the conformal boundary $\partial M$, and the normal plane corresponds to the radial and $\s^k$ bulk dimensions. A null geodesic $\gamma$ reaches $\partial M$ with future-directed tangent vector $N$. The unit-norm orthogonal projection of this vector $\hat{N}^\perp$ gives a point on $\sk$ which corresponds to the asymptotic location of $\gamma$ on the compact space. }}
    \label{fig:bnb}
\end{figure}

\subsection{Definition of extended light-cone cuts}\label{ssec:geocuts}

In order to recover the higher-dimensional bulk conformal metric, we will need the point on $\sk$ associated with the null geodesic going from $p$ to $C(p)$. Since $\sk$ shrinks to zero size on the boundary, we will use the $\epsilon\to0$ limit of $\partial_{\epsilon(\Omega_k^\infty)}$. The latter can be characterized geometrically as the unit vector $\hat{N}^\perp\in\R^{k+1}$ along the projection of the null tangent vector $N$ orthogonal to $\partial M$. As a unit vector in a $(k+1)$-dimensional vector space, one can identify $\hat{N}^\perp$ with a point on $\sk$, the coordinates of which are $\Omega_k^\infty$ (see Fig. \ref{fig:bnb}).

 Unfortunately, more than one null geodesic may connect $p$ to $C(p)$, so the assignment of a point on $\sk$ may not be unique.
{This motivates the following definition:}

\begin{ndefi}[Regular light-cone cut point]\label{def:regular}
	A regular light-cone cut point $r\in C(p)$ for some bulk point $p\in M$, is a cut point such that there exists a unique null geodesic from $p$ to $r$.
\end{ndefi}

It is tempting to think of a point $r$ that fails to be regular as belonging to some caustic on the light cone, as is the case in spacetimes without internal spaces. While this will commonly be true here too, one should bear in mind that the null geodesics that connect $p$ and $r$ might actually stay at finite proper distance apart on $\sk$, only coinciding strictly at the conformal boundary. If this happened to be the case for all null geodesics connecting $p$ to $r$, these points would not be conjugate points, and thus it would not be correct to think of $r$ as arising from some bulk caustic. To account for this subtlety, it will be useful to dispense with the notion of caustics and use only what happens to be relevant from the boundary perspective in identifying whether a cut point is regular. Two null vectors $N_1$ and $N_2$ at $r$ clearly define inequivalent null geodesics if and only if one is not a rescaling of the other. Hence the failure of a light-cone cut point $r\in C(p)$ to be regular can be characterized by the existence of at least two null geodesics $\gamma_1$ and $\gamma_2$ from $p$ to $r$ with respective tangent vector fields $N_1$ and $N_2$ satisfying $(N_1\cdot N_2)_r\neq0$. It will thus be intuitive to refer to a non-regular cut point as a \emph{cusp point}.

Let $G(p) \subseteq C(p)$ be the subset of regular points in the light-cone cut of $p\in M$. On this subset, there exists a well-defined map $\Phi : G(p) \to \sk$ associating a point on the unit $k$-sphere to every regular point. Explicitly, as remarked above, this map may be written
\begin{equation}\label{eq:skmap}
    \Phi(r) = \hat{N}_r^\perp,
\end{equation}
where an isomorphism between the unit $\sk$ and the space of $(k+1)$-dimensional unit vectors is implied (see Fig. \ref{fig:elcc}). In contrast, there is no guarantee that an analogous map on the set of cusp points $E(p) \equiv C(p) \smallsetminus G(p)$ would be well-defined due to potential multi-valuedness on $\sk$.

\begin{ndefi}[Extended light-cone cut]\label{defi:extcut}
	The extended future/past light-cone cut $\mathcal{C}^\pm(p)$ of a point $p\in M$ is defined {on the set of regular points $G^\pm(p)\subseteq C^\pm(p)$} as $$ \mathcal{C}^\pm(p) = \bigcup_{r\in{G}^\pm(p)} \(r,\Phi(r)\).$$
\end{ndefi}
\noindent These extended cuts $\mathcal{C}(p)$ may be thought of as a generalization of the standard cuts $C(p)$ where every suitable point, namely every $r\in G(p)$, is further endowed with the point on $\sk$ at which the null geodesic from $p$ to $r$ ends up (see Fig. \ref{fig:elcc}).

\begin{figure}
    \centering
    \includegraphics[width=.9\textwidth]{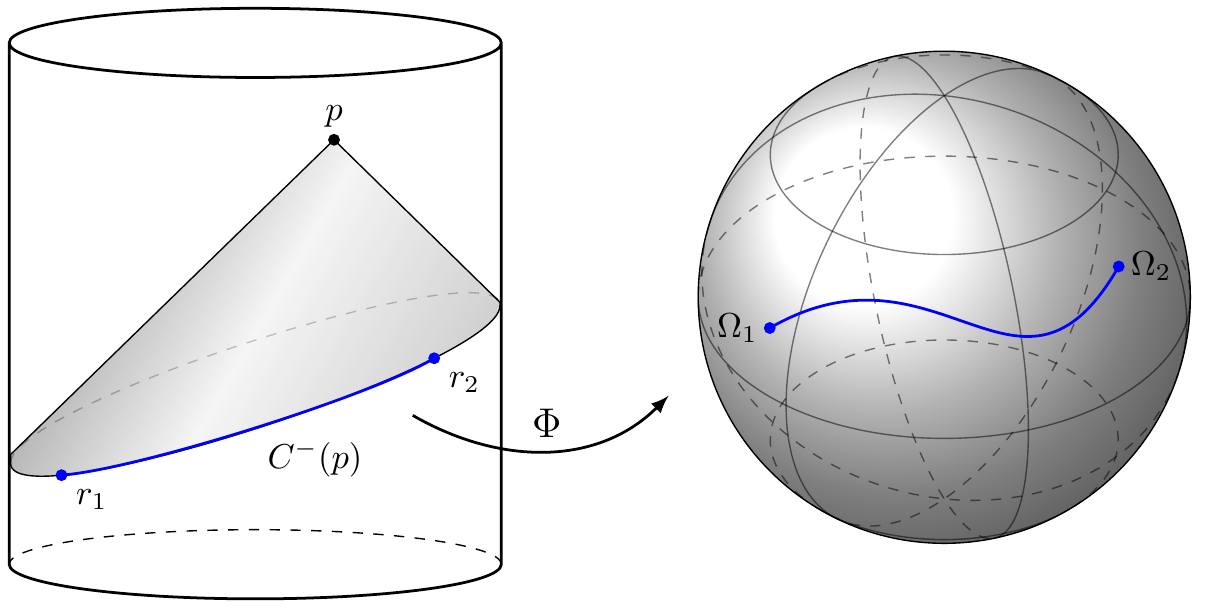}
    \caption{{The map $\Phi$ defined in \eqref{eq:skmap} takes (regular) points in the cut $C(p)$ of a point $p\in M$ and maps them to $\sk$. For instance, the blue segment between points {$r_1,r_2\in C^-(p)$} maps to the blue segment between points $\Omega_1,\Omega_2\in\sk$.}}
    \label{fig:elcc}
\end{figure}

Since the reconstruction strategy relies on the existence of regular points on which the map \eqref{eq:skmap} is defined, it is important to check whether $G(p)$ contains sufficiently many points at one's disposal. An important step in this direction is accomplished by the following proposition, which as proven in Appendix \ref{sec:maths} and applies to light-cone cuts in spacetimes with degenerate boundaries:
\begin{restatable}[]{proposition}{propzeromes}
\label{prop:zeromes}
	Every light-cone cut $C(p)$ is differentiable everywhere except on a set of measure zero.
\end{restatable}

A cut $C(p)$ can be non-differentiable at $r$ only if there is more than one null geodesic from $p$ to $r$. So any point $r$ at which $C(p)$ fails to be differentiable will be a cusp point $r\in E(p)$, and thus the set of all non-differentiable cusp points is of measure zero in $C(p)$.

Since there may be cusp points where $C(p)$ is differentiable, this is not enough to conclude anything about the measure of $E(p)\subseteq C(p)$. However, {differentiability at cusp points is only possible if all geodesics from $p$ to $r$ happen to have tangent vectors at $r$ with the} same normalized projection onto $\partial M$. Fortunately, given one vector, {such a} condition on the second is satisfied only by a set of measure zero and thus the set of all differentiable cusp points is expected to be of measure zero in $E(p)$.

Putting together the conclusions of the last two paragraphs, one expects that the union of all differentiable and non-differentiable cusp points, which is nothing but the set of all cusp points $E(p)$, is of Lebesgue measure zero as a subset of $C(p)$. This implies that its complement, i.e. the set of all regular points $G(p)$, is of full measure, everywhere dense and that its closure $\bar{G}(p)=C(p)$.

The key property of the extended cut that we will use is the following: 
\begin{restatable}[]{proposition}{propnormalgeo}
\label{prop:normalgeo}
    Each point $(r,\Phi(r))$ on the extended cut $\mathcal{C}(p)$ determines the unique null geodesic from $r$ to $p$.
\end{restatable}

\noindent This result is proven in Appendix \ref{sec:maths} and provides the connection to previous results in \cite{Engelhardt2017,Engelhardt2016}.

\subsection{Recovering the bulk conformal metric from extended cuts}\label{ssec:oldwork}

The following results apply to standard light-cone cuts $C(p)$ and their proofs are identical to those in \cite{Engelhardt2016}, so they are omitted:\footnote{These results correspond to parts $(1)$ and $(2)$ of the Proposition in \cite{Engelhardt2016}.}
\begin{nprop}
	 $C(p)$ is a complete spatial slice of $\partial M$.
\end{nprop}

\begin{nprop}\label{prop:causalbdy}
	For any $p\in J^\pm[\partial M]$, there exists precisely one past/future cut $C^\mp(p)$.
\end{nprop}

The following results, in contrast, are generalizations of results in \cite{Engelhardt2016} which now apply to extended light-cone cuts (see Appendix \ref{sec:maths} for proofs):\footnote{These results are analogous to (a stronger version of) part $(3)$ of the Proposition and Theorem $1$ in \cite{Engelhardt2016}.}

\begin{restatable}[]{proposition}{propmorethanonepoint}
\label{prop:morethanonepoint}
	$\mathcal{C}(p)\cap \mathcal{C}(q)$ contains more than one point if and only if $p=q$.
\end{restatable}

\begin{nthm}\label{thm:exactlyonepoint}
	If $\mathcal{C}(p)\cap \mathcal{C}(q)$ contains exactly one point, then $p$ and $q$ are distinct and null-related.
\end{nthm}

Actually, a slightly stronger version of Theorem \ref{thm:exactlyonepoint} is proven in Appendix \ref{sec:maths}. The idea of the proof is simply that the common point on both extended cuts defines an ingoing null geodesic that must go through both $p$ and $q$, and hence they must be null related.

\begin{figure}
    \centering
    \includegraphics[width=.9\textwidth]{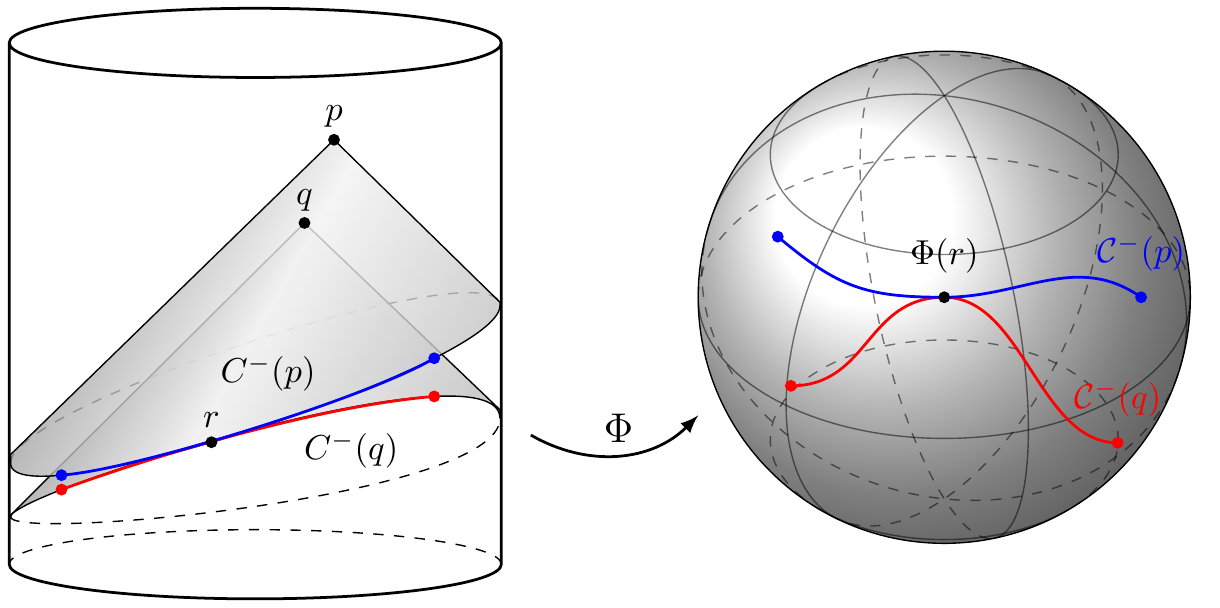}
    \caption{By Theorem \ref{thm:exactlyonepoint}, two extended cuts $\mathcal{C}^-(p)$ and $\mathcal{C}^-(q)$ intersect at precisely one point $(r,\Phi(r))$ only if they correspond to distinct, null-related points $p,q\in M$. As illustrated, this requires both that their standard cuts $C(p)$ and $C(q)$ intersect at precisely one point $r\in\partial M$, and also that their images under the map $\Phi$ intersect precisely at $\Phi(r)\in\sk$.}
    \label{fig:tec}
\end{figure}

From Proposition \ref{prop:causalbdy}, the extended past cuts provide a copy of the space $J^+[\partial M]$. From Theorem \ref{thm:exactlyonepoint}, we can determine a class of null directions at each point $p\in J^+[\partial M]$, by looking for extended cuts $\mathcal{C}^-(q)$ which intersect $\mathcal{C}^-(p)$ at precisely one point. This situation is depicted in Fig. \ref{fig:tec}. One cannot recover all null directions at $p$ but only those corresponding to null geodesics which stay on $\partial J^-(p)$. To obtain the conformal metric, one needs a basis of null directions. So the key question is whether the null directions we can reconstruct form a basis. This is not obvious since the tangent space at $p$ is $n+k$ dimensional, and we only have access to an $n-2$ dimensional space of null directions associated to points of $\mathcal{C}(p)$. {For instance, the answer would be negative in} a spacetime which is globally a product such as $\adss$, since $\mathcal{C}(p)$ would yield null geodesics that are everywhere orthogonal to $\sk$. However, for a generic spacetime without any symmetries acting only on the internal space asymptotically, one expects the $n-2$ dimensional space of null geodesics corresponding to $\mathcal{C}(p)$ to span the tangent space, and not all be orthogonal to {any} vector at $p$. Hence we expect that one can generically {construct} a basis of null vectors $\ell_i$ at $p$. One can then choose the additional null vectors $\eta_k$ and determine the conformal metric as described in Section \ref{sec:review}.


\section{Data from the Dual Field Theory}\label{sec:bdydata}

In the present context, there is no obstruction to obtaining the standard light-cone cuts from the perspective of the boundary theory via the method reviewed in Section \ref{sec:review} and originally presented in \cite{Engelhardt2017,Engelhardt2016}.\footnote{Notice, though, that the required number of correlator insertions to obtain a bulk-point singularity now needs to account for the bulk dimensions, not the boundary dimensions. In other words, one needs at least $n+k+1$ operators, not just $n+1$. See Section \ref{ssec:bpsing} for more details.} Nevertheless, as observed in Section \ref{sec:exlcc}, knowledge of the cuts $C(p)$ is not sufficient for the reconstruction of the higher-dimensional bulk metric when the latter has a degenerate conformal boundary. The additional information needed for such reconstruction to succeed is encoded in the extended light-cone cuts $\mathcal{C}(p)$ and given by the map $\Phi : G(p) \to \sk$ from regular cut points to the asymptotic $k$-sphere. The main focus of this section is to address the problem of how to obtain this extra ingredient solely from the dual field theory. We will propose a  procedure to recover this map to the sphere using only field theory correlators.

\subsection{Higher-dimensional bulk-to-boundary propagator}\label{ssec:bpsource}

From the bulk perspective, the action that describes some matter field $\varphi$ is naturally defined on all $D=n+k$ dimensions of the bulk spacetime. Accordingly, the bulk-to-bulk propagator $\mathcal{G}$ takes as input the coordinates $X$ of bulk points in some higher-dimensional chart, i.e. $X\in\R^{D}$.\footnote{In what follows, it suffices to work with retarded and advanced propagators. Under an appropriate notion of global or AdS hyperbolicity, these are well-defined and unique in general time-dependent spacetimes \cite{bar2007wave}.} In particular, if $\varphi$ obeys an equation of motion of the form $\mathds{P}_X \varphi(X) = J(X)$ for some source term $J$, then $\mathcal{G}$ is defined as the Green function of $\mathds{P}_X$,
\begin{equation}
    \mathds{P}_X \mathcal{G}(X,\tilde{X}) = \frac{1}{\sqrt{\det g}} \delta^{D}(X-\tilde{X}).
\end{equation}
Although it is a natural object, $\mathcal{G}$ rarely appears in the literature (see \cite{Dorn_2005} for an exception in global $\adss$). Instead, propagators are commonly obtained after dimensionally-reducing spacetime and Kaluza-Klein expanding on the compact dimensions. The result is an infinite family of simpler propagators associated to the infinite Kaluza-Klein tower of modes which, holographically, correspond to operators of definite conformal dimension. However, in a completely general spacetime, there is no well-defined way of discriminating the compact dimensions far from the conformal boundary. Hence, one cannot hope to learn much about the higher-dimensional spacetime geometry from the perspective of boundary correlators unless one understands how all such Kaluza-Klein mode propagators combine into the higher-dimensional propagator $\mathcal{G}$ and its bulk-to-boundary analogue $\mathcal{K}$. The goal of this section is to define and understand these higher-dimensional propagators and demonstrate how they may be used to obtain the map $\Phi$ for the construction of the extended cuts. 

Although $\mathcal{G}$ is a perfectly well-defined object, it turns out to be nontrivial to obtain an explicit, compact expression for it for a general minimally-coupled Klein-Gordon scalar field even in global $\adss$. Without simplifying assumptions, the latter can be expressed as an infinite Kaluza-Klein series expansion as in \eqref{eq:Gsolution}. However for a conformally flat choice of radii in $\adss$ and a specific mass term for the scalar corresponding to the Weyl invariant coupling to the scalar curvature, this infinite sum can be recast into the very simple form of \eqref{eq:compactbu2bu} \cite{Dorn_2005}. One example of this is the massless dilaton in $\text{AdS}_5\times \s^5$ (with equal radii), since the scalar curvature vanishes.

On the other hand, the bulk-to-boundary propagator $\mathcal{K}$ is a more subtle object. For local $\adss$ asymptotics, a canonical choice of coordinates near the conformal boundary is Fefferman-Graham $(z,x)$ on the $\ads$ part \cite{fefferman_1985} and standard hyperspherical coordinates $\Omega$ on the $\sk$ part. Accordingly, in some neighborhood of the conformal boundary one may set $X=(z,x,\Omega)$. Despite $\Omega$ being a degenerate coordinate for any point on the conformal boundary, corresponding to $z=0$, the limiting value of $\Omega$ exists along some curves as $z\to0$ (cf. the discussion at the beginning of Section \ref{ssec:geocuts}). From this standpoint, one would expect that some generalization of the extrapolate dictionary should allow one to obtain the bulk-to-boundary propagator $\mathcal{K}$ given the bulk-to-bulk propagator $\mathcal{G}$. In particular, one would hope to construct an object like $\mathcal{K}(\tilde{X} ;\, x,\Omega)$ as some limit $z\to0$ along curves of constant $(x,\Omega)$ of $\mathcal{G}(\tilde{X} ;\, z,x,\Omega)$, where $\tilde{X}$ are the coordinates of an arbitrary bulk point. 

There is a subtlety, though: because the boundary is a conformal boundary, one generally only considers the $z\to0$ limit of propagators of definite scaling dimension, for which it is clear which power of $z$ the leading term carries. Asymptotically, this scaling dimension is associated to Kaluza-Klein modes arising from the dimensional reduction of the $\sk$. But by virtue of being higher-dimensional, the propagator $\mathcal{G}$ incorporates all such modes, and therefore the extrapolation of it to $\mathcal{K}$ via the $z\to0$ limit should take care of all of them at once. Due to these complications, we shall take a more axiomatic approach in defining $\mathcal{K}$.

As a bulk-to-boundary propagator, $\mathcal{K}$ will be defined to be the kernel of $\mathds{P}_X$, i.e. the solution to the homogeneous equation
\begin{equation}\label{eq:Kdef}
    \mathds{P}_{\tilde{X}} \mathcal{K}(\tilde{X} ;\, x,\Omega) = 0,
\end{equation}
and subject to some choice of boundary conditions at $\partial M$. These conditions are imposed on the limit in which the bulk point approaches the conformal boundary too. In this limit, $\tilde{X}=(\tilde{z},\tilde{x},\tilde{\Omega})$ is again an appropriate chart and as $\tilde{z}\to0$ one may work with the intuition that $\adss$ provides. In particular, by dimensionally reducing near the conformal boundary, one can decompose $\mathcal{K}$ into contributions from lower-dimensional propagators for all possible Kaluza-Klein modes $\mathcal{K}_{\Delta}$ of definite scaling dimension $\Delta$. Thus, at least for $\tilde X$ near the boundary, we have
\begin{equation}\label{eq:Ksum}
    \mathcal{K}(\tilde{X} ;\, x,\Omega) = \sum_{L=0}^{\infty} \mathcal{K}_{\Delta_L}(\tilde{X} ;\, x,\Omega).
\end{equation}
The dependence on $\sk$ is not arbitrary, but fixed by the choice of boundary conditions. For the usual Dirichlet conditions one would demand that $\lim_{\tilde{z}\to0} \tilde{z}^{\Delta-d} \mathcal{K}_\Delta(\tilde{z},\tilde{x},\tilde{\Omega} ;\, x,\Omega) \propto \delta^d(x-\tilde{x})$, where $d=n-1$. In the higher-dimensional analogue, the $\sk$ coordinates really correspond to physical, compact dimensions, and the Dirichlet conditions should be imposed on those too. This motivates accounting for all Kaluza-Klein modes $L$ in the definition of boundary conditions via
\begin{equation}\label{eq:kernelbdy}
    \lim_{\tilde{z}\to0} \; \sum_{L=0}^{\infty} \tilde{z}^{\Delta_L-d} \mathcal{K}_{\Delta_L}(\tilde{z},\tilde{x},\tilde{\Omega} ;\, x,\Omega) = {\frac{1}{\sqrt{\det g_{\sk}}}} \delta^d(x-\tilde{x}) \delta^k(\Omega-\tilde{\Omega}).
\end{equation}

We can now continue the propagator $\mathcal{K}$ deeper inside the bulk as a kernel of $\mathds{P}_{X}$ using either retarded or advanced evolution. The result is our desired bulk-to-boundary propagator in the full spacetime. This approach is followed in {Appendix \ref{assec:bu2by}} to obtain the general form of the bulk-to-boundary propagator for the Klein-Gordon scalar field in global $\adss$, expressed as an infinite series in \eqref{eq:Kformal} (cf. the bulk-to-bulk series in \eqref{eq:Gsolution}).\footnote{Since this spacetime is static, Euclidean propagators are used in Appendix \ref{asec:props}.} In the particular case of Weyl-invariant matter, it is again possible to resum this series expansion and obtain a compact expression, namely \eqref{eq:explicitK}.

\subsection{The compact space from the dual field theory}\label{ssec:carto}

The asymptotic form of a scalar field $\varphi$ on an asymptotically locally $\adss$ spacetime admits a Kaluza-Klein expansion over the $\sk$ in scalar hyperspherical harmonics $Y_{L}^{I_L}$ of the form\footnote{For more details on how the harmonic functions $Y_L^{I_L}$ are defined see Appendix \ref{assec:bu2bu}.}
\begin{equation}
\varphi(z,x,\Omega ) = \sum_{L=0}^\infty \sum_{I_L} Y_{L}^{I_L}(\Omega) \varphi_{L}^{I_L}(z,x).
\end{equation}
According to the holographic dictionary, the leading asymptotic term of the non-normalizable branch of every mode
\begin{equation}
\lim_{z\to0} z^{\Delta_L-d} \varphi_{L}^{I_L}(z,x) = \phi_{L}^{I_L}(x),
\end{equation}
becomes a source of a local boundary operator $\mathcal{O}_{L}^{I_L}(x)$ of definite conformal dimension $\Delta_L$. Introducing a generic bulk field $\varphi$ involving arbitrarily many Kaluza-Klein modes thus corresponds to turning on arbitrarily heavy operators on the boundary theory. Explicitly, the bulk partition function  is equal to a field theory partition function involving a complicated operator sum $\mathcal{O}_\phi$ of the form
\begin{equation}\label{eq:phiop}
\mathcal{O}_\phi(x) = \sum_{L=0}^\infty \sum_{I_L} \phi_{L}^{I_L}(x) \mathcal{O}_{L}^{I_L}(x).
\end{equation}
As a boundary operator in its own right, $\mathcal{O}_\phi$ creates a bulk field with a conformal asymptotic profile $\phi(x,\Omega)$ which is given by contributions from all sources
\begin{equation}\label{eq:boundarysource}
\phi(x,\Omega) = \sum_{L=0}^\infty \sum_{I_L} Y_{L}^{I_L}(\Omega) \phi_{L}^{I_L}(x).
\end{equation}
Following this intuition and using a quantum mechanical language, at any fixed boundary coordinate $x$, the insertion of $\mathcal{O}_\phi(x)$ produces a particle which is thrown into the bulk localized at a point in $\partial M$ with coordinates $x$ and whose wavefunction is spread over the asymptotic $\sk$ according to $\phi(x,\Omega)$ as a function of $\Omega$. More explicitly, the action of the operator $\mathcal{O}_\phi(x)$ on the vacuum state $\Ket{0}$ of the boundary theory creates a state $\Ket{\phi_x} = \mathcal{O}_\phi(x) \Ket{0}$. When projected onto the position basis $\Omega$ of $\sk$, this state reads $\braket{\Omega}{\phi_x} = \phi(x,\Omega)$, whereas when projected onto the basis of eigenfunctions $Y_L^{I_L}$ of $\square_{\sk}$, it reads $\langle{Y_L^{I_L}}\mid{\phi_x}\rangle = \phi_{L}^{I_L}(x)$.

Consider the following object, a generalization of which will be relevant in the next subsection:
\begin{equation}\label{eq:bulkbdy2pt}
\Pi(\tilde{X};\,x) = \int_{\sk} d\Omega \; \phi(x,\Omega) \mathcal{K}(\tilde{X} ;\,x,\Omega).
\end{equation}
For instance, in global $\adss$, using the Kaluza-Klein expanded form of $\mathcal{K}$ in \eqref{eq:Kformal},
\begin{equation}
\Pi(\tilde{z},\tilde{x},\tilde{\Omega};\,x) = \sum_{L=0}^\infty \sum_{I_L} {Y_L^{I_L}}(\tilde{\Omega}) \; \phi_{L}^{I_L}(x) K_{\Delta_L}(\tilde{z},\tilde{x} ;\,x),
\end{equation}
where $K_{\Delta_L}$ is the usual $L$-mode bulk-to-boundary propagator, given in \eqref{eq:KLonly}. 
The bilocal field $\Pi$ in \eqref{eq:bulkbdy2pt} can be thought of as the response function of a boundary probe $\phi$ at $x$ smeared over the $\sk$ to a localized bulk source at $\tilde{X}$ propagated through spacetime via $\mathcal{K}$. This interpretation will naturally follow from a more complicated but closely related construct in Section \ref{ssec:bpsing} that comes out of a correlation function which boundary observers have access to. Although the right-hand side of \eqref{eq:bulkbdy2pt} is integrated over $\Omega$, note that $\Pi$ depends on the profile of $\phi$ as a function of $\Omega$ and is thus sensitive to dependencies on the asymptotic $\sk$. More precisely, if a boundary observer who can measure $\Pi$ had complete control over $\phi$, by tuning the boundary profile to be $\phi(x,\Omega)=\delta^k(\Omega-\Omega')$ parametrized by $\Omega'$, it would be possible for them to scan over $\Omega'$ and reproduce $\mathcal{K}$ precisely. However, note that by completeness of the spherical harmonics, such a choice of $\phi$ would correspond to picking $\phi_L^{I_L}(x) = {Y_L^{I_L}}^*(\Omega')$, which according to \eqref{eq:phiop} would build $\mathcal{O}_\phi$ out of operators $\mathcal{O}_L^{I_L}$ of all dimensions $L$, including arbitrarily heavy ones.

More realistically, one might want to only use light operators and get as good an approximation to $\mathcal{K}$ as possible. With this goal, consider letting $\phi_L^{I_L}=\delta_{L\tilde{L}} \delta^{I_L {\tilde{I}_{\tilde{L}}}}$ in \eqref{eq:boundarysource} (which corresponds to simply $\mathcal{O}_\phi = \mathcal{O}_{\tilde{L}}^{{\tilde{I}_{\tilde{L}}}}${), and label the resulting right-hand side in \eqref{eq:bulkbdy2pt} by $\Pi_{\tilde{L}}^{{\tilde{I}_{\tilde{L}}}}$}. This allows one to invert \eqref{eq:bulkbdy2pt} by writing $\mathcal{K}$ as a harmonic series
\begin{equation}\label{eq:Kseries}
    \mathcal{K}(\tilde{X} ;\,x,\Omega) = \sum_{L=0}^\infty \sum_{I_L} \, {\Pi_L^{I_L}(\tilde{X};\,x)} \, {Y_L^{I_L}}^*(\Omega),
\end{equation}
where the correlators in the sum are effectively the Fourier coefficients of the expansion. For an approximation to $\mathcal{K}$, one may want to employ $L$ modes only up to some finite cut-off $L_\infty<\infty$. It should be noted that \eqref{eq:Kseries} applies to any asymptotically locally $\adss$ spacetime (cf. \eqref{eq:Ksum} and comments below).

We would like to obtain the position of a local bulk source solely from the boundary perspective using the bulk-to-boundary propagator. It is pertinent at this point to make clear the semantic distinction between \emph{localizing} and \emph{locating}. We do not want to create a perturbation \emph{localized} on $\sk$, which would require the whole tower of Kaluza-Klein modes. Instead, what we want is to \emph{locate} a source that already is localized on $\sk$, which need not require such high-$L$ physics. Indeed, in the tractable case of global $\adss$, we now show that using $\Pi$ it is possible from the boundary perspective to find the exact location on $\sk$ of a localized bulk source employing just $L=1$ operators.

Let $\mathcal{O}_\phi$ only involve light operators in the fundamental representation of $SO(k+1)$ such that only $L=1$ harmonics contribute to $\phi$. With homogeneous sources, a general expression for the latter is obtained by writing the coefficients $\phi_L^{I_L}=\delta_{L,1} {Y_1^{I_1}}^*(\Omega)$ parameterized by a point $\Omega$ on $\sk$. 
Suggestively writing $\Pi(\tilde{X} ;\, x) = \Pi_1(\tilde{X} ;\, x,\Omega)$ for this choice of $\phi$, 
\eqref{eq:bulkbdy2pt} becomes\footnote{Note that $\Omega$ here has been introduced as just a parameter for the choice of coefficients $\phi_L^{I_L}$.}
\begin{equation}\label{eq:L1prop}
\begin{aligned}
\Pi_1(\tilde{z},\tilde{x},\tilde{\Omega} ;\, x,\Omega ) & = \sum_{I_1} {Y_1^{I_1}}^*(\Omega) \int_{\sk} d\Omega' \; Y_{1}^{I_1}(\Omega') \mathcal{K}(\tilde{X} ;\,x,\Omega') \\ & = (k+1) K_{\Delta_1}(\tilde{z},\tilde{x};\,x) \cos\theta,
\end{aligned}
\end{equation}
where $\theta$ is the angular separation between coordinates $\tilde{\Omega}$ and $\Omega$ on $\sk$. Therefore, a boundary observer that is able to vary $\Omega$ will find $\Omega=\tilde{\Omega}$ precisely at the maximum of $\Pi_1$, corresponding to $\theta=0$. This shows that, from the boundary perspective, the function $\Pi_1$ of $L=1$ modes allows one to \emph{locate} the exact position on $\sk$ at which a \emph{localized} bulk source resides.

\subsection{Recovering the extended cut}\label{ssec:bpsing}

As observed in previous sections, even with a compact space asymptotically, one can determine the standard light-cone cuts $C(p)$ from bulk-point singularities in certain boundary correlators. The only change is the number of operators in the correlator. The light cone of an arbitrary boundary point permeates the bulk as a submanifold of bulk-codimension one. In a generic spacetime, the intersection of  the light cones of $\ell$ arbitrary boundary points will generically be a submanifold of bulk-codimension $\ell$ (or the empty set when $\ell>D$.)\footnote{As stated, it is important for this result to be generic that the spacetime compactification does not factorize exactly or have exact symmetries, and that the boundary points be chosen arbitrarily. Global $\adss$ is thus clearly non-generic.} So one needs at least $D$ operators {to single out a point in the bulk}. {In this section we} further refine the usage of these correlators in order to obtain the map $\Phi$. In particular, the strategy will be to find $\Phi$ from the prefactor of the {leading} divergent {term} of bulk-point singular correlators, which {exhibits a suitable dependence} on the asymptotic $\sk$.

We start with a divergent correlator {as} used to find the standard light-cone cuts $C^\pm(p)$ of some bulk point $p\in M$,\footnote{Any set of operators referred to henceforth shall be assumed to correspond to some local interaction term in the action of the bulk theory. For example, in $10$-dimensional supergravity, there is a coupling between the dilaton and $3$-form, $e^{-2\phi} H_3^2$. Expanding out the exponential yields $\phi^{10} H_3^2$ interaction terms.}

\begin{equation}\label{eq:ptsing}
\left\langle \mathcal{T}\left\{ \prod_{i=1}^{D+2} \mathcal{O}(x_i) \right\} \right\rangle
\end{equation}
where $\mathcal{T}$ denotes time ordering.
As argued above, the choice of any $D$ such points $x_i$ in the correlator above singles out $p$ as the unique bulk point that is null-related to all of them. As shown in Fig. \ref{fig:ebpc}, we place these $D$ points on the future cut and add two points on the past cut. By moving these two points in a way that keeps the correlator divergent (which requires maintaining momentum conservation at $p$), we can trace out the past cut. This does not depend on the choice of operator insertions. Let us denote all but one of these boundary points collectively by $\vec{x}=\{ x_i\in C^\pm(p) \st i = 2,\dots, D+2\}$ and use $L=0$ scalar operators $\mathcal{O}$ at all of these $D+1$ points.\footnote{One may want to consider more general insertions $\mathcal{O}_{\phi_i}$ at each boundary point, but this is unnecessary.} To obtain the extended cut, it will be convenient to work with a probe point $x$ near the remaining point $x_1$. We choose an operator $\mathcal{O}_\phi$ at $x$ as in \eqref{eq:phiop} which is sensitive to the $\s^k$.

\begin{figure}
    \centering
    \includegraphics[height=.4\textheight]{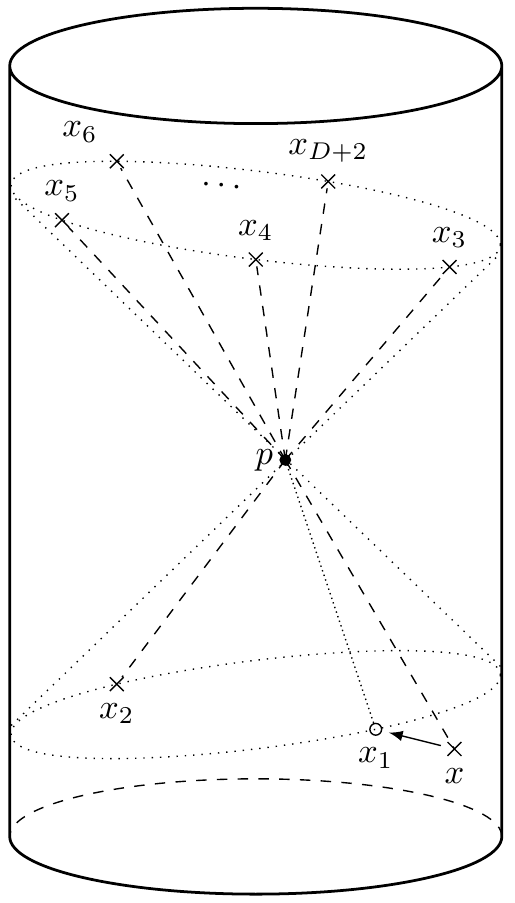}
    \caption{{Configuration of boundary points in the correlator \eqref{eq:fphi} used to obtain the extended light-cone cuts from the dual field theory. By choosing suitable operators at $x$ and looking at the coefficient of the divergence when $x$ approaches the cut point $x_1\in C^-(p)$, one can obtain the map $\Phi$ from regular cut points to the $\sk$.}}
    \label{fig:ebpc}
\end{figure}

As a result of the existence of the null-related, momentum-preserving point $p$, the $(D+2)$-point, time-ordered Lorentzian correlator
\begin{equation}\label{eq:fphi}
F_{\phi}(x) \equiv \left\langle \mathcal{T}\left\{ \mathcal{O}_{\phi}(x) \prod_{i=2}^{D+2} \mathcal{O}(x_i) \right\} \right\rangle
\end{equation}
will develop a bulk-point singular contribution in the limit $x\to x_1$ \cite{Maldacena2015}. This divergent limit of interest is illustrated in Fig. \ref{fig:ebpc}. Written out in a particularly convenient form, for some choice of normalization, the correlation function \eqref{eq:fphi} reads
\begin{equation}\label{eq:convcorr}
F_{\phi}(x) = \int_{\sk} d\Omega \; \phi(x,\Omega) \int_{M} d^{D}\tilde{X} \; \Psi_{\vec{x}}(x,\Omega;\,\tilde{X}),
\end{equation}
where the integrand of the bulk-point integral is 
\begin{equation}\label{eq:psix}
\Psi_{\vec{x}}(x,\Omega;\,\tilde{X}) = \lambda \; \mathcal{K}(x,\Omega ;\, \tilde{X}) \int_{\sk} \prod_{i=2}^{D+2} d\Omega_i \; \mathcal{K}(x_i,\Omega_i ;\, \tilde{X}),
\end{equation}
and $\lambda$ is the coupling of some local interaction involving the $D+2$ fields of interest.

The dominant bulk-point singular contribution from $p$ to \eqref{eq:convcorr} manifests itself as the highest-order pole in $\Psi_{\vec{x}}$, precisely located at {the coordinates $X_p$ of $p$}, close to which the function $\Psi_{\vec{x}}$ will be governed by a power-law divergence in the proper distance between $\tilde{X}$ and $X_p$. To intuitively see why this is the case, observe first that propagators generally behave as inverse powers of proper distances between the points in their arguments, here with coordinates $(x_i,\Omega_i)$ and $\tilde{X}$. Importantly, because the asymptotic $\sk$ trivializes on $\partial M$, this proper distance does not depend on the value of $\Omega_i$ for the boundary point. Now, since all $x_i$ are null-related to $X_p$, for $\tilde{X}$ in a small neighborhood of $X_p$, to leading order the proper distance $s(x_i,\tilde{X})$ between any boundary point $x_i$ and $\tilde{X}$ will be proportional to $s(\tilde{X},X_p) \approx \norm{\tilde{X}-X_p}$, where the use of the Minkowski metric in the last approximation is justified by local flatness at $X_p$. The dependence on the choice of boundary points $\vec{x}$ is thus relegated simply to the specification of the unique bulk point $p$ in this equation (cf. the rank argument in \cite{Maldacena2015}) and the form of the residue of the pole of $\Psi_{\vec{x}}$ at $X_p$. The order of the dominant pole $\Delta_{D+2}$ depends on the operator insertions and details of the spacetime metric.\footnote{In the case of global $\adss$, the symmetries lead to $\Delta_{D+2}$ being just a sum over the largest scaling dimension of each of the boundary operator insertions.} Pulling out the leading divergent factor in $\Psi_{\vec{x}}$, one may write
\begin{equation}\label{eq:leadingPsi}
\Psi_{\vec{x}}(x,\Omega;\,\tilde{X}) = \frac{\psi(x,\Omega;\,\tilde{X})}{\norm{\tilde{X}-X_p}^{\Delta_{D+2}}},
\end{equation}
where now the function $\psi$ is finite and non-zero at $\tilde{X}=X_p$. To leading order in the distance $\norm{x-x_1}$ off the light-cone cut, the integral of \eqref{eq:leadingPsi} over $\tilde{X}$ will be dominated by the zeroth order term of $\psi$ in a series expansion about $\tilde{X}=X_p$ and evaluated at $x=x_1$. This leads to 
\begin{equation}\label{eq:importantresult}
F_{\phi}(x) = I(x) \int_{\sk} d\Omega \; \phi(x,\Omega) \psi(x_1,\Omega ;\, X_p),
\end{equation}
where $I(x)$ captures the bulk-point singularity as $x\to x_1$ from the integral over $\tilde{X}$,\footnote{The order of the pole agrees with the result in \cite{Maldacena2015} if one identifies $\Delta_{D+2} \to (D+1)\Delta$ (corresponding to $D+1$ external vertices rather than $D+2$), and $D \to d + 1$ (corresponding to no internal space).}
\begin{equation}
I(x) \propto \norm{x-x_1}^{-\( \Delta_{D+2} - D \)}.
\end{equation}

The previous section showed that in global $\adss$ it was possible to locate the unique direction specified by $\Omega_1$ in which the null geodesic from $X_p$ arrives at $x_1$ {using the object defined in \eqref{eq:bulkbdy2pt}}. The reason for this could be traced back to the fact that the higher-dimensional propagator $\mathcal{K}$ in \eqref{eq:L1prop} had a global maximum at $\Omega=\Omega_1$. By causality, this fact is expected to extend to arbitrary spacetimes, where now the general function $\psi$ in \eqref{eq:importantresult} is the object peaked at $\Omega = \Omega_1$.\footnote{In $\adss$, it sufficed to use light modes with $L=1$ to locate this point since in this highly symmetric case, all nontrivial Kaluza-Klein modes are peaked at the same point. In a general spacetime, it is still expected that $\psi$ will have a global maximum at $\Omega = \Omega_1$, but no single $L$ mode need be peaked there. Put differently, $\psi$ will generically exhibit no symmetries in $\Omega$ and higher $L$ will be required to locate $\Omega_1$.}

In analogy to the previous section, choosing $\phi$ to to consist of a single $Y_L^{I_L}$, define
\begin{equation}\label{eq:FdefiniteL}
F_{L}^{I_L}(x) = I(x) \int_{\sk} d\Omega \; {Y_L^{I_L}}(\Omega) \psi(x_1,\Omega ;\, X_p),
\end{equation}
which, up to $I(x)$, may be thought of as the Fourier coefficients of an expansion of $\psi$ into hyperspherical harmonic functions. Inverting this relation leads to
\begin{equation}\label{eq:invertedFdefiniteL}
\psi(x_1,\Omega ;\, X_p) = \lim_{x\to x_1} \frac{\psi_0}{F_0(x)} \sum_{L=0}^{\infty} \sum_{I_L} F_{L}^{I_L}(x) {Y_L^{I_L}}^*(\Omega),
\end{equation}
where $F_0$ corresponds to \eqref{eq:FdefiniteL} for $L=0$ and is introduced to cancel out the common bulk-point singular factor of every term in the series. The constant $\psi_0$, given by
\begin{equation}\label{eq:psirecover}
    \psi_0 = \frac{1}{\vol\s^k} \int_{\s^k} d\Omega \; \psi(x_1,\Omega;\, X_p),
\end{equation}
is irrelevant and will be left undetermined.\footnote{If $\psi_0$ vanishes identically so will $F_0$, and one may just use a different $L$ mode to cancel out singular factors.} The upshot is that, up to an overall constant, $\psi$ can be reconstructed to arbitrary precision by computing the terms in the series in \eqref{eq:invertedFdefiniteL} for increasingly high $L$ values. Since the right-hand side is built solely out of boundary correlators, this information is in principle accessible to boundary observers. Once obtained, the location of the global maximum of $\psi$ in $\Omega$, namely $\Omega_1$, determines the desired map $\Phi$ to the asymptotic $\sk$. More explicitly, one obtains $\Phi(x_1)=\Omega_1$ from the solution to $\psi(x_1,\Omega;\, X_p) = \max \psi$, where the specific choice of the $D+1$ additional light-cone cut points $\vec{x}$ may be ignored since it is arbitrary so long as they belong to the same choice of past and future light-cone cuts $C(p)$ of point $p$ at $X_p$ (see Fig. \ref{fig:ebpc}).


\section{Discussion}\label{sec:discuss}
Most discussions of bulk reconstruction in holography consider asymptotically AdS spacetimes and ignore the extra compact directions. This was true for the discussion of light-cone cuts in \cite{Engelhardt2016,Engelhardt2017}. 
We have considered asymptotically locally $\adss$ spacetimes and defined a generalization of light-cone cuts that we call extended cuts. We then showed that in the region of spacetime causally connected to the boundary, one can generically recover the full higher-dimensional conformal metric just from the location of the extended cuts. Finally, we proposed a procedure for determining these extended cuts from the dual field theory. Note that at no time did we need to use any bulk equations of motion, {or impose any restrictions on the matter content (such as energy conditions)}.

Our proposal for determining the extended cuts from the dual field theory is not very practical since it requires considering the entire tower of Kaluza-Klein modes to precisely locate the bulk points. But the lesson is that the information is there in principle.\footnote{In practice, from the perspective of the dual field theory, one would probably first want to know how many extra dimensions the bulk spacetime has. This interesting question was recently addressed in \cite{Alday:2019qrf}.} It would be interesting to find a more efficient way to determine the extended cuts.

Although we have focused on the case where the bulk metric asymptotically approaches $\adss$, our reconstruction should work equally well for spacetimes that asymptotically approach $\ads\times K$, where $K$ is any compact Einstein space. To see this, note that null geodesics that remain on $\partial J(p)$ will again approach a fixed point on $K$, and we can again define our extended cut to be the light-cone cut $C(p)$ together with a map $C(p) \to K$. The arguments in Section \ref{sec:exlcc} then apply to show that the conformal metric can be reconstructed from the location of these extended cuts. One difference with $\sk$ is that when we conformally rescale the asymptotic metric, the result will take the form \eqref{eq:barg} with $d\Omega^2$ replaced by the metric on $K$, which will be singular at the conformal boundary. This should not be a problem since our arguments only require that each point on the extended cut defines a unique ingoing null geodesic in the bulk. Since we know the geodesic starts at a fixed point on $K$, and the bulk metric reduces to pure AdS asymptotically when a point on $K$ is held fixed, the geodesic leaves the boundary exactly as it would in pure AdS. The arguments in Section \ref{sec:bdydata} also extend to this case since the hyperspherical harmonics on $\sk$ can be replaced by the eigenfunctions of the Laplacian on $K$ which form a complete basis of functions. Scalar fields can be expanded in terms of these functions, yielding the usual infinite tower of massive Kaluza-Klein modes in the asymptotic AdS region. Holography requires that there is a CFT operator dual to each of these modes which we can use in our correlators. 

It is natural to ask how quantum or stringy corrections affect our arguments. It was argued in \cite{Maldacena2015} that bulk-point singularities would still be present when perturbative $1/N$ or $1/\lambda$ corrections to holography are included, but not for finite $N$ or $\lambda$. More recently, the stringy resolution of these singularities has been quantified under some general assumptions in \cite{Dodelson2019}. Since bulk-point singularities are a key ingredient in our approach, we note that exact reconstruction of the bulk conformal metric is possible with perturbative but not finite quantum or stringy corrections. 

We close with a few open questions. First, to recover the full bulk metric and not just the conformal metric, we clearly need a procedure to obtain the conformal factor. One would like this to be independent of the bulk equations of motion. Second, general arguments on bulk reconstruction \cite{Dong:2016eik} show that one should be able to reconstruct the higher-dimensional metric on the entire entanglement wedge of the boundary. The light-cone cut approach to bulk reconstruction only applies to points in the causal wedge, since they have to be in causal contact with the boundary both to the past and future. Actually, not all points in the causal wedge are accessible since momentum must be conserved near the vertex. This means that points just outside the horizon of an eternal black hole are excluded since they are causally connected to infinity only through a narrow cone. 

To expand the reach of the light-cone cuts, we either need more general ways to obtain the extended cuts from the dual field theory, or we need to use other methods (perhaps combined with light-cone cuts). The recent work of \cite{May:2019odp} relating bulk scattering and holographic entanglement suggests a plausible direction to connect the light-cone cut approach to bulk reconstruction with those based on entropic measures, thereby hinting at a potentially synergistic combination of the two. It is nevertheless important to note that it is unlikely that the standard holographic entanglement entropy as given by the prescriptions in \cite{Ryu_2006,Hubeny_2007} could on its own be used for higher-dimensional bulk reconstruction. In a variety of nontrivial examples, it has been shown that these prescriptions applied to dimensionally-reduced spacetimes give entropies that agree with those obtained by performing the extremization problem on the full higher-dimensional spacetime, which suggest that the latter carries no more information about the extra dimensions than the former does \cite{Jones_2016}. Intuitively, this is a consequence of the boundary condition that instructs the higher-dimensional extremal surfaces to wrap uniformly around the compact dimensions asymptotically. However one might be able to generalize these ideas, perhaps along the lines of \cite{Mollabashi_2014,Karch_2015,Taylor_2016}, to probe the higher-dimensional geometry. If a suitable boundary interpretation of this generalized entropy is available, one could perhaps use e.g. some upgraded version of the arguments in \cite{Bao_2019} to prove uniqueness of the higher-dimensional metric and potentially come up with an entropy-based reconstruction strategy.

\acknowledgments
We thank X. Dong, S. Fischetti, J. Maldacena, D. Marolf and H. Maxfield for helpful discussions.
This work was supported in part by the National Science Foundation under Grants PHY-1801805 and PHY-1748958.
SHC was also supported by fellowship number LCF/BQ/AA17/11610002 from ``la Caixa'' Foundation (ID 100010434).

\appendix

\section{Mathematical Results}\label{sec:maths}
In this appendix we give the proofs of the new results stated in section 3. We will assume that the spacetime $M$ is at least $C^2$, maximally extended, connected, AdS-hyperbolic and asymptotically locally $\adss$ with $n>2$. Similarly, $\partial M$ is assumed to be maximally extended, connected, and globally hyperbolic. Recall that AdS-hyperbolic means that there are no closed causal curves, and for any two points $p$ and $q$, the set $J^+(p)\cap J^-(q)$ is compact after conformally compactifying the AdS boundary \cite{Wall_2014}.

We will assume everywhere in this section that $p$ and $q$ are bulk points in the domain of influence of the asymptotic boundary, so {that their} light-cone cuts are not empty. The results below apply to both future and past light-cone cuts which we denote $C(p)$, or $\mathcal{C}(p)$ for the extended cuts.
In expressions like
$\mathcal{C}(p) \cap \mathcal{C}(q)$ it will be understood that both cuts are past or both cuts are future.

\propzeromes*
\begin{proof}
	The logic of the first part of this proof parallels that of Proposition $6.3.1$ in \cite{Hawking1973}.\footnote{There is a typo in the proof in \cite{Hawking1973}: both instances of the set $\mathscr{L}$ appearing in the penultimate sentence should be replaced by its boundary set $\dot{\mathscr{L}}$.} Let $r\in C(p)$ and consider an open neighborhood $U_\alpha\subset\partial M$ about $r$. One can introduce normal coordinates $x_\alpha=\{x_\alpha^\mu : U_\alpha \to\R \st \mu = 0,\dots,n-2\} $ with $\partial_{0}$ timelike and such that the coordinate slices $\gamma_{\vec{c}} = \{s\in U_\alpha \st \vec{x}_\alpha(s) = \vec{c} \}$, where $\vec{x}_\alpha = \{x_\alpha^i \st i=1,\dots,n-2\}$, define curves intersecting both $I^-(r)\cap U_\alpha$ and $I^+(r)\cap U_\alpha$ for any constant $\vec{c} \in \vec{x}_\alpha[U_\alpha]$. By continuity and achronality, each curve $\gamma_{\vec{c}}$ must intersect $C(p)$ at precisely one point $s_{\vec{c}}$, i.e. $\{s_{\vec{c}}\}=\gamma_{\vec{c}}\cap C(p)$, and therefore the map $\vec{x}_\alpha : U_\alpha\cap C(p) \to \R^{n-2}$ is a homeomorphism onto its image.
	
	Now define a map $\tilde{x}^0 : \vec{x}[U_\alpha\cap C(p)] \to x^0[U_\alpha]$ by $\tilde{x}^0(\vec{c}) = x^0(r_{\vec{c}})$, where $r_{\vec{c}}$ is the unique point at which $\gamma_{\vec{c}}$ intersects $C(p)$. Because $C(p)$ is achronal, for any two points $r,s\in U_\alpha\cap C^{}(p)$ one has that $\abs{\tilde{x}^0(\vec{x}(r)) - \tilde{x}^0(\vec{x}(s))} \leq K \abs{\vec{x}(r)-\vec{x}(s)}$ for some $K\geq1$, with $\abs{\,\cdot\,}$ the Euclidean norm. This shows that $\tilde{x}^0$ is Lipschitz continuous. A Lipschitz continuous transition map $\varphi_{\alpha\beta} : \vec{x}_\alpha[U_\alpha \cap U_\beta \cap C(p)] \to \vec{x}_\beta[U_\alpha \cap U_\beta \cap C(p)]$ can now be constructed by direct product and composition with maps of higher differentiability class as $\varphi_{\alpha\beta} = \vec{x}_\alpha \circ x_\beta^{-1} \circ \{\tilde{x}^0_\beta,\,\id\}$. Thus a collection of charts $(U_\alpha\cap C(p),\,\vec{x}_\alpha)$ forms an atlas for $C(p)$ and endows it with a Lipschitz structure.
	
	The differentiability of $C(p)$ at a point $r \in U_\alpha \cap C(p)$ is determined by the differentiability class of the transition maps $\varphi_{\alpha\beta}$ at $\vec{x}_\alpha(r)\in\vec{x}_\alpha[U_\alpha \cap U_\beta \cap C(p)]$. Because the transition map $\varphi_{\alpha\beta}$ is Lipschitz continuous, Rademacher's theorem \cite{RademacherberPU} implies that the points in $\vec{x}_\alpha[U_\alpha \cap U_\beta \cap C(p)]\subset \R^{n-2}$ at which $\varphi_{\alpha\beta}$ is not differentiable form a set of Lebesgue measure zero as a subset of $\R^{n-2}$. Thus the set of points at which $C(p)$ fails to be differentiable has measure zero.
\end{proof}

\propnormalgeo*
\begin{proof}
    Consider an arbitrary point $(r,\Phi(r))\in\mathcal{C}(p)$ and let $\gamma:[0,1]\to\bar{M}$ be the unique null geodesic from $r=\gamma(0)\in C(p)$ to $p=\gamma(1)\in M$. Write $\dot{\gamma}(0) \propto V^\perp + V^\parallel$ with $V$ some rescaled vector parallel to $\dot{\gamma}(0)$ such that $V^\perp$ has unit norm, where $V^\perp$ ($V^\parallel$) is the projection of $V$ onto the normal (tangent) bundle of $\partial M$. Since $r$ is regular, $C(p)$ is differentiable at $r$, and therefore there is a well-defined space tangent to $C(p)$ at $r$, denoted $T_r C(p)$. Because $C(p)$ is a codimension-$1$ spacelike subspace of $\partial M$, there is a unique timelike vector $T\in T_r \partial M$ normal to $C(p)$ with $T^2=-1$. Under natural identifications of the vectors in $\partial M$ with their inclusions in the ambient space $\bar{M}$, one can further decompose $\dot{\gamma}(0) \propto T + \cos\alpha \, V^\perp + \sin\alpha \, S$, where $S\in T_r C(p)$ is a unit spacelike vector and $\alpha\in[0,\pi/2)$. If $\alpha\neq0$, there would be a nontrivial vector $S$ such that one could consider a point $r_\epsilon\in C(p)$ arbitrarily close to $r$ in the direction parallel to $S$. Notice that then one could deform $\gamma$ infinitesimally near $\partial M$ into a timelike piece that connects up with $r_\epsilon$, thus making $p$ and $r_\epsilon$ timelike-related, which contradicts the achronality of the light cone $\partial J(p)$. Hence one finds that $\alpha=0$ necessarily, and therefore $\dot{\gamma}(0) \propto T + V^\perp$ in general. In other words, regularity of the cut point implies that the vector field $\dot{\gamma}$ tangent to $\gamma$ is orthogonal to $C(p)$ at $r$. The dimensionality of the normal bundle of $C(p)$ in $\bar{M}$ is given by the codimension of $C(p)$, which is $k+2$ corresponding to timelike and radial bulk directions and the conformally-shrinking $\sk$. The specification of $\Phi(r)$ by the extended cut point fixes the direction of $\dot{\gamma}(0)$ on $\sk$, such that this vector only remains undetermined in $2$ dimensions. Out of the $2$ possible null directions spanning the latter, only one points inwards towards the bulk. Hence the choice of a point in $\mathcal{C}(p)$ together with the orthogonal ingoing condition fix $\dot{\gamma}(0)$ up to scaling. But because $\dot{\gamma}(0)$ is null and $\gamma$ is geodesic, this suffices to determine a unique null geodesic from $r$ to $p$.
\end{proof}

\propmorethanonepoint*
\begin{proof}
    ~\\
	$(\Leftarrow)$ If $p=q$, $\mathcal{C}(p)\cap \mathcal{C}(q)=\mathcal{C}(p)$, which always contains more than one point. \\
	$(\Rightarrow)$ Consider an arbitrary point $(r,\Phi(r))\in \mathcal{C}(p)\cap \mathcal{C}(q)$. According to Proposition \ref{prop:normalgeo}, the pair $(r,\Phi(r))$ determine a unique ingoing null geodesic $\gamma_r$. If there were two distinct such points in the intersection of the two extended cuts, their associated $\gamma_r$ geodesics would pass through both $p$ and $q$, which would then be either equal or conjugate to each other. But since there cannot be any conjugate points along any $\gamma_r$ strictly between either of these points and their cuts, it must be the case that $p=q$.
\end{proof}

\begin{nthm}\label{thm:stronger}
	$\mathcal{C}(p) \cap \mathcal{C}(q)$ contains exactly one point $(r,\Phi(r))$ if and only if $q\neq p$ and $q$ belongs to an achronal extension of a null geodesic $\gamma$ from $p$ to a regular point $r\in C(p)$.\footnote{The statement of an analogous result in \cite{Engelhardt2017} is not quite correct. In particular, $q$ need not belong to the null geodesic from $p$ to $r\in C^\pm(p)$, but instead could lie on an extension of this geodesic beyond $p$ (i.e. $p$ itself would lie in a null geodesic connecting $q$ to $r$). This explains the qualification of the statement to an \emph{achronal extension} of the null geodesic from $p$ to $r$.}
\end{nthm}
\begin{proof}
    ~\\
    $(\Rightarrow)$\footnote{This direction proves Theorem \ref{thm:exactlyonepoint} in Section \ref{ssec:oldwork}.} Since $\mathcal{C}(p)\cap \mathcal{C}(q)$ contains exactly one point $(r,\Phi(r))$, Proposition \ref{prop:morethanonepoint} already implies $p\neq q$. Then Proposition \ref{prop:normalgeo} shows that $(r,\Phi(r))$ defines the unique null geodesic $\gamma_r$ associated to the regular point $r$. Since $(r,\Phi(r))$ belongs to the intersection of the two cuts, $\gamma_r$ passes through both $p$ and $q$ and stays on the union of their light cones $\partial J(p)\cup \partial J(q)$. Thus the two points are null-related by an achronal geodesic through both that ends at $r$. \\
	$(\Leftarrow)$ If $p$ and $q$ both lie on an achronal null geodesic $\gamma$ that reaches a regular point $r\in C(p)$ and $\Phi(r)$ on $\sk$, then $\gamma$ lies on both $\partial J(p)$ and $\partial J(q)$. So $(r,\Phi(r))$ is clearly in both extended cuts $\mathcal{C}(p)$ and $\mathcal{C}(q)$.
\end{proof}


\section{Higher-dimensional scalar propagators in global $\adss$}\label{asec:props}

Consider a free bulk scalar field of mass $m$ with Euclidean action
\begin{equation}\label{eq:matterS}
S_\varphi = \frac{1}{2} \int \epsilon \( \abs{d\varphi}_g^2 + m^2 \varphi^2 \),
\end{equation}
where $\epsilon$ is the volume element on all $D=n+k$ dimensions of $\adss$. Using Poincar\'{e} coordinates in Euclidean signature,
\begin{equation}\label{eq:euclideanmetric}
g = \frac{\ell^2}{z^2} \( dz^2 + \delta_{ij} dx^i dx^j \) + \ell^2 d\Omega^2,
\end{equation}
where Latin indices run over the $d=n-1$ spatial dimensions of AdS$_n$.

\subsection{Bulk-to-bulk propagator}\label{assec:bu2bu}
The higher-dimensional bulk-to-bulk scalar propagator $\mathcal{G}$ is defined as the Green function of the Klein-Gordon operator,
\begin{equation}\label{eq:Gdefinition}
\( - \square_g + m^2 \) \mathcal{G}(z,x,\Omega ;\, \tilde{z},\tilde{x},\tilde{\Omega}) = \frac{1}{\sqrt{\det g}} \delta^{n}(z-\tilde{z},x-\tilde{x})\delta^k(\Omega-\tilde{\Omega}).
\end{equation}
where $\square_g$ denotes the d'Alembertian built from the $D$-dimensional metric $g$. Because $\adss$ is a product spacetime, this operator is diagonal and decomposes as
\begin{equation}\label{eq:dalembert}
\square_g = \square_{\ads} + \ell^{-2} \Delta_{\sk},
\end{equation}
where $\square_{\ads}$ only acts on AdS coordinates $(z,x)$ and the unit $k$-sphere Laplacian $\Delta_{\sk}$ only acts on coordinates $\Omega$. Explicitly,
\begin{equation}
\square_{\ads} = \frac{z^2}{\ell^2} \( \partial_z^2 - (d-1) z^{-1} \partial_z + \partial_x^2 \),
\end{equation}
and, using Cartesian coordinates on $\R^{k+1}\supset\sk$, one can write
\begin{equation}
\Delta_{\sk} = \sum_{\alpha>\beta}^k \( x_\alpha \partial_\beta - x_\beta \partial_\alpha \)^2.
\end{equation}

Consider first the propagator $G_\Delta$ of a free scalar in AdS$_n$ of mass $\mu$, defined by
\begin{equation}\label{eq:usualGeq}
\( - \square_{\ads} + \mu^2 \) G_\Delta(z,x ;\, \tilde{z},\tilde{x}) = \frac{1}{\sqrt{\det g_{\ads}}} \delta^n(z - \tilde{z}, x - \tilde{x}).
\end{equation}
This Green function is well-known and can be written in terms of the hypergeometric function $_2 F_1$ as \cite{Gubser:1998bc,DHoker2002}
\begin{equation}\label{eq:greeng}
G_\Delta(z,x ;\, \tilde{z},\tilde{x}) = \frac{2^{-{\Delta}} C_{\Delta}}{2{\Delta}-d} \, \xi^{\Delta} \, {_2 F_1}\(\frac{{\Delta}}{2} ,\, \frac{{\Delta}}{2} + \frac{1}{2} ;\, {\Delta} - \frac{d}{2} + 1 ;\, \xi^2\),
\end{equation}
where the conformal ratio $\xi$ is defined in terms of the coordinates of the two points by
\begin{equation}
\xi \equiv \frac{2 z \tilde{z}}{z^2 + \tilde{z}^2 + (x - \tilde{x})^2},
\end{equation}
and the conformal dimension $\Delta$ and normalization constant $C_{\Delta}$ are
\begin{equation}\label{eq:clml}
\mu^2 = \frac{{\Delta}({\Delta} - d)}{\ell^2} \nd C_{\Delta} = \frac{\Gamma{({\Delta})}}{\pi^{d/2} \Gamma({\Delta} - d/2)}.
\end{equation}
The two solutions of the quadratic equation obeyed by $\Delta$ correspond to the usual two branches
\begin{equation}\label{eq:deltal}
\Delta = \frac{d}{2} \pm \sqrt{\frac{d^2}{4} + \ell^2 \mu^2},
\end{equation}
with the positive (negative) sign giving the normalizable (non-normalizable) one.

Consider now the $\sk$ term in \eqref{eq:dalembert}. The eigenfunctions of $\Delta_{\sk}$ are called hyperspherical harmonics $Y_L^I(\Omega)$ and labeled by their scaling degree $L\in\Z_{\geq0}$ and a tuple $I_L=(i_1,\dots,i_{k+1})\in\mathbb{Z}^{k+1}$ with $\sum_{l=1}^{k+1} i_l = L$, which specifies an element of the representation of $SO(k+1)$ in terms of traceless symmetric tensors of degree $L$ in $k+1$ dimensions. They are defined by the eigenvalue problem
\begin{equation}\label{eq:orthonormality}
\Delta_{\sk} Y_L^{I_L}(\Omega) = - L(L+k-1) Y_L^{I_L}(\Omega).
\end{equation}
and conventionally orthonormalized to satisfy
\begin{equation}
\int_{\sk} d\Omega \, {Y_L^{I_L}}^*(\Omega) \, Y_{\tilde{L}}^{\tilde{I}_{\tilde{L}}}(\Omega) = \delta_{L \tilde{L}} \delta^{I_L {\tilde{I}_{\tilde{L}}}},
\end{equation}
where $d\Omega$ is the volume element of $\sk$. Additionally, as a basis for functions on $\sk$, hyperspherical harmonics obey the completeness relation
\begin{equation}\label{eq:completeness}
\sum_{L=0}^{\infty} \sum_{I_L} \, {Y_L^{I_L}}^*(\tilde{\Omega}) \, Y_L^{I_L}(\Omega) = \frac{1}{\sqrt{\det g_{\sk}}} \delta^k(\Omega - \tilde{\Omega}).
\end{equation}
The sum over $SO(k+1)$ representation indices $I_L$ for fixed $L$ can be performed explicitly and leads to \cite{Wen1985}
\begin{equation}\label{eq:ILsum}
\sum_{I_L} {Y_{L}^{I_L}}^*(\Omega)\, Y_{L}^{I_L}(\tilde{\Omega}) = N_L \, C_L^{(k-1)/2} \( \cos\theta \),
\end{equation}
where $\cos(\theta) \equiv \vec{n} \cdot \tilde{\vec{n}}$ for unit vectors $\vec{n},\tilde{\vec{n}}\in\R^{k+1}$ oriented on $\sk$ as specified by $\Omega$ and $\tilde{\Omega}$, respectively, and $N_L$ is a normalization constant given by
\begin{equation}\label{eq:nlvs}
N_L = \frac{2L+k-1}{(k-1) \vol \sk} \where \vol\s^{2l-1} = \frac{2\pi^l}{\Gamma(l)},
\end{equation}
The symbol $C^\alpha_l(x)$ is a Gegenbauer polynomial, which can be written as
\begin{equation}
C^\alpha_l(x) = \frac{\Gamma(2\alpha + l)}{\Gamma(2\alpha)} \, {_2 F_1} \( -l ,\, 2\alpha + l ;\, \alpha + \frac{1}{2} ;\, \frac{1-x}{2} \).
\end{equation}

It is now a simple matter to construct the desired propagator:

\begin{nprop}
    The higher-dimensional bulk-to-bulk propagator $\mathcal{G}$ for a free scalar of mass $m$ in global $\adss$ given as an infinite series by
    \begin{equation}\label{eq:Gsolution}
    \mathcal{G}(z,x,\Omega ;\, \tilde{z},\tilde{x},\tilde{\Omega}) = \sum_{L=0}^{\infty} N_L C_L^{(k-1)/2} \(\cos\theta \) G_{\Delta_L}(z,x ;\, \tilde{z},\tilde{x}),
    \end{equation}
    where $G_{\Delta_L}$ is the propagator of a free scalar in $\ads$ of scaling dimension $\Delta_L$ defined to be
    \begin{equation}\label{eq:massL}
        \Delta_L = \frac{d}{2} \pm \sqrt{\frac{d^2}{4} + \ell^2 M_L^2} \where M_L^2 = m^2 + \frac{L(L+k-1)}{\ell^2}.
    \end{equation}
\end{nprop}

\begin{proof}
    Applying the right-hand side of \eqref{eq:Gdefinition} to \eqref{eq:Gsolution} leads to
    \begin{equation}\label{eq:extendedsquare}
        \begin{aligned}
            \(-\square_g+m^2 \) \mathcal{G}(z,x,\Omega ;\, \tilde{z},\tilde{x},\tilde{\Omega}) & = \sum_{L=0}^\infty N_L C_L^{(k-1)/2}(\cos\theta) \( -\square_{\ads} + M_L^2 \) G_{\Delta_L}(z,x ;\, \tilde{z},\tilde{x}) \\
            & = \frac{1}{\sqrt{\det g}} \delta^{n}(z-\tilde{z},x-\tilde{x})\delta^k(\Omega-\tilde{\Omega})
        \end{aligned}
    \end{equation}
    where \eqref{eq:dalembert}, \eqref{eq:ILsum}, \eqref{eq:orthonormality} and \eqref{eq:massL} have been used in the first equality, and \eqref{eq:usualGeq}, \eqref{eq:ILsum} and \eqref{eq:completeness} in the second one. The result thus agrees with the right-hand side of \eqref{eq:Gdefinition}.
\end{proof}

The series form of \eqref{eq:Gsolution} may be understood as a Kaluza-Klein series expansion of the higher-dimensional bulk-to-bulk propagator. This expression reduces to a very compact form for conformally flat $\adss$, as is the case of \eqref{eq:euclideanmetric},\footnote{Recall that global $\adss$ is conformally flat if and only if the radius of the $\sk$ matches that of $\ads$.} if one chooses the scalar to be coupled to the metric in a Weyl invariant manner \cite{Dorn_2005}. This is accomplished in \eqref{eq:matterS} by choosing the mass of the scalar to be precisely
\begin{equation}\label{eq:weylmass}
    m^2 = \frac{(k-1)^2 - (n-1)^2}{4 \ell^2}.
\end{equation}
The resulting propagator is simply a power-law in the total chordal distance along both $\ads$ and $\sk$, viz. (see \cite{Dorn_2005} for more details)
\begin{equation}\label{eq:compactbu2bu}
    \mathcal{G}(z,x,\Omega ;\, \tilde{z},\tilde{x},\tilde{\Omega}) = \frac{\Gamma(h)}{2 (2\pi)^{h+1}} \frac{1}{(\xi^{-1} - \cos\theta)^{h}} \where h=\frac{n+k-2}{2}.
\end{equation}

\subsection{Bulk-to-boundary propagator}\label{assec:bu2by}

One would naively hope to be able to derive a simple expression for the bulk-to-boundary propagator starting from \eqref{eq:compactbu2bu} and using some version of the extrapolate dictionary \cite{Harlow:2011ke}. Unfortunately, it is not at all clear in this case how one would take the $z\to0$ limit of \eqref{eq:compactbu2bu}. A naively reasonable guess would be to Taylor expand this object in $\xi$, kill off the $z^l$ power in $\xi^l$ of the $l^\text{th}$ term with a factor of $(2 l - d) z^{-l}$, take the $z\to0$ limit and hope to be able to perform the summation of the resulting series to obtain a compact expression. However, this would neither be a kernel as defined in \eqref{eq:Kdef} nor obey the desired boundary condition in \eqref{eq:kernelbdy}.

Instead, our approach will be to perform the summation in \eqref{eq:Ksum} directly. The terms in the summation can be obtained by applying the extrapolate dictionary to every term in the series \eqref{eq:Gsolution} that defines $\mathcal{G}$. These will involve the usual dimension-$\Delta_L$ bulk-to-boundary propagator \cite{Witten1998,Harlow:2011ke}
\begin{equation}\label{eq:KLonly}
K_{\Delta_L}(z,x ;\,\tilde{x}) = \lim_{\tilde{z}\to0} \, (2\Delta_L - d) \tilde{z}^{-\hat{\Delta_L}} \; G_{\Delta_L}(z,x ;\, \tilde{z},\tilde{x}) = C_{\Delta_L} \chi^{\Delta_L},
\end{equation}
where $\chi$ is given by
\begin{equation}
    \chi = \frac{z}{z^2 + (x-\tilde{x})^2}.
\end{equation}
The upshot is the following infinite Kaluza-Klein series definition of the higher-dimensional bulk-to-boundary propagator:
\begin{equation}\label{eq:Kformal}
    \mathcal{K}(z,x,\Omega;\,\tilde{x},\tilde{\Omega}) = \sum_{L=0}^{\infty} N_L C_{\Delta_L} C_L^{(k-1)/2} \(\cos\theta \) \chi^{\Delta_L}.
\end{equation}
Note that in Lorentz signature, the limit that the bulk and boundary point become null separated corresponds to $\chi \to \infty$. Each term in this series then develops a singularity with a coefficient that is a smooth function on $\s^k$ peaked at the location of the bulk point.

The computation of this sum becomes tractable for Weyl invariant matter, which fixes the mass of the scalar to be given by \eqref{eq:weylmass}. The resulting $L^{\text{th}}$ term in \eqref{eq:Kformal} is
\begin{equation}\label{eq:KLterm}
    \mathcal{K}_{\Delta_L}(z,x,\Omega;\,\tilde{x},\tilde{\Omega}) = \frac{\Gamma \(\frac{k-1}{2}\)}{2 \pi^{h + 1}} \( L + \frac{k-1}{2} \) \frac{\Gamma(L+h)}{\Gamma \(L+\frac{k-1}{2}\)} C_L^{(k-1)/2}(\cos\theta) \; \chi^{L+h},
\end{equation}
where $h$ was defined in \eqref{eq:compactbu2bu}. For convenience, focus on the odd-$n$ case, for which the ratio of $\Gamma$ functions may be expanded as a finite product. Using the Pochhammer symbol $(a)_n=a(a+1)\cdots(a+n-1)$, this is
\begin{equation}\label{eq:prod}
    \frac{\Gamma(L+h)}{\Gamma \(L+\frac{k-1}{2}\)} = \(L + \frac{k-1}{2}\)_{\frac{n-1}{2}}.
\end{equation}
The goal will be to manipulate \eqref{eq:KLterm} so as to be able to utilize the identity of Gegenbauer polynomials that gives their defining generating function, namely\footnote{Note that this identity holds as an equality between power series in $\chi$. However, as an infinite series, the left-hand side is only convergent for $\abs{\chi}<1$. While this should be kept in mind, in practice in will not be a problem: physically, one is interested in looking at each $L$ mode independently. Every relation derived henceforth using this identity should thus be understood as an equality between power series in $\chi$.}
\begin{equation}\label{eq:genfun}
    \sum_{L=0}^{\infty} C_L^{\alpha}(y) \, \chi^L = \frac{1}{\(1-2 \chi y + \chi^2\)^\alpha}.
\end{equation}
To do this, note that the right-hand side of \eqref{eq:prod} can be realized via differentiation in $\chi$ in the following way
\begin{equation}\label{eq:diff1}
    \frac{\Gamma(L+h)}{\Gamma \(L+\frac{k-1}{2}\)} = \frac{1}{\chi^{L+\frac{k-1}{2}-1}} \partial_\chi^{\frac{n-1}{2}} \(\chi^{L+h-1}\),
\end{equation}
and similarly one can write
\begin{equation}\label{eq:diff2}
    L+\frac{k-1}{2} = \frac{1}{\chi^{L+\frac{k-1}{2}-1}}\partial_\chi \( \chi^{L+\frac{k-1}{2}} \).
\end{equation}
Putting \eqref{eq:diff1} and \eqref{eq:diff2} together with $\chi^{L+h}$, consider the following manipulations:
\begin{equation}
\begin{aligned}
    \(L+\frac{k-1}{2}\) \frac{\Gamma(L+h)}{\Gamma \(L+\frac{k-1}{2}\)} \chi^{L+h}
    & = \frac{1}{\chi^{L+\frac{k-1}{2}-1}} \partial_\chi \( \chi^{L+\frac{k-1}{2}} \frac{\Gamma(L+h)}{\Gamma \(L+\frac{k-1}{2}\)} \)\; \chi^{L+h} \\
    & = \chi^{\frac{n+1}{2}} \partial_\chi \( \chi \, \partial_\chi^{\frac{n-1}{2}} \(\chi^{L+h-1}\) \).
\end{aligned}
\end{equation}
With this expression at hand, the infinite series that defines $\mathcal{K}$ may now be rewritten as
\begin{equation}\label{eq:exactK}
    \sum_{L=0}^{\infty} \mathcal{K}_{\Delta_L}(z,x,\Omega;\,\tilde{x},\tilde{\Omega}) = \frac{\Gamma \(\frac{k-1}{2}\)}{2 \pi^{h + 1}} \chi^{\frac{n+1}{2}} \partial_\chi \( \chi \, \partial_\chi^{\frac{n-1}{2}} \( \chi^{h-1} \sum_{L=0}^{\infty} C_L^{(k-1)/2}(\cos\theta) \, \chi^{L}\) \).
\end{equation}
At this point it only remains to employ \eqref{eq:genfun} to obtain the desired explicit form of the higher-dimensional bulk-to-boundary propagator:
\begin{equation}\label{eq:explicitK}
    \mathcal{K}(z,x,\Omega;\,\tilde{x},\tilde{\Omega})= \frac{\Gamma \(\frac{k-1}{2}\)}{2 \pi^{h + 1}} \chi^{\frac{n+1}{2}} \partial_\chi \( \chi \, \partial_\chi^{\frac{n-1}{2}} \( \frac{\chi^{h-1}}{\(1-2 \chi \cos\theta + \chi^2\)^{\frac{k-1}{2}}}\) \).
\end{equation}
This result is valid for any odd $n\geq3$ and any integer $k\geq 2$. Note also that this expression only holds as an equality between coefficients in a power series in $\chi$, the reason being that the radius of convergence of the infinite series is $\abs{\chi}<1$. This is not a problem in Euclidean signature because $\abs{\chi}<1$ always, but should be kept in mind for Lorentzian signature where e.g. null separation corresponds to $\chi\to\infty$.

The spacetime with one of the simplest evaluations of \eqref{eq:explicitK} is $\text{AdS}_3\times\s^3$, for which
\begin{equation}
    \mathcal{K}(z,x,\Omega;\,\tilde{x},\tilde{\Omega}) = 
    \frac{\chi^2 \left(\chi^4+2 \chi \(\chi^2+1\) \cos\theta -6 \chi^2+1\right)}{2 \pi ^3 \left(\chi^2-2 \chi \cos\theta+1\right)^3}.
\end{equation}
For the usual case of interest of $\text{AdS}^5\times\s^5$ one gets the following:
\begin{equation}
    \mathcal{K}(z,x,\Omega;\,\tilde{x},\tilde{\Omega}) = -\frac{2 \chi^4 \(\chi^6+2 \chi^4 \cos 2\theta -17 \chi^4+25 \chi^2+2 \left(4 \chi^4-5 \chi^2-3\right) \chi \cos\theta-3\right)}{\pi ^5 \left(\chi^2-2 \chi \cos\theta+1\right)^5}.
\end{equation}

\addcontentsline{toc}{section}{References}
\bibliographystyle{JHEP}
\bibliography{references}

\end{document}